\documentclass[a4paper,11pt]{article}

\usepackage{fullpage}
\usepackage{times}
\usepackage{soul}
\usepackage{url}
\usepackage[utf8]{inputenc}
\usepackage[small]{caption}
\usepackage{graphicx}
\usepackage{xcolor}
\usepackage{amsmath}
\usepackage{booktabs}
\usepackage{algorithm}
\usepackage{algorithmic}
\usepackage{enumitem}

\usepackage{amsthm}
\usepackage{amssymb}

\usepackage{bbm}

\newtheorem{theorem}{Theorem}[section]
\newtheorem{corollary}[theorem]{Corollary}

\newtheorem{lemma}[theorem]{Lemma}

\theoremstyle{definition}

\usepackage{natbib}

\newcount\Comments  
\newcommand{\kibitz}[2]{\ifnum\Comments=1{\color{#1}{#2}}\fi}

\newcommand{\vv}{\mathbf{v}}
\newcommand{\SW}{\text{SW}}
\newcommand{\EQ}{\text{EQ}}
\newcommand{\PoA}{\text{PoA}}
\newcommand{\PoS}{\text{PoS}}
\newcommand{\calI}{\mathcal{I}}

\allowdisplaybreaks

\usepackage{authblk}

\title{\bf Schelling Games on Graphs\thanks{This work has been supported by the European Research Council (ERC) under grant number 639945 (ACCORD), and by the KAKENHI Grant-in-Aid for JSPS Fellows number 18J00997.}}

\author[1]{Edith Elkind}
\author[1]{Jiarui Gan}
\author[2]{Ayumi Igarashi}
\author[1]{\\ Warut Suksompong}
\author[1]{Alexandros A. Voudouris}

\affil[1]{Department of Computer Science, University of Oxford}
\affil[2]{Department of Computer Science and Communication Engineering, Kyushu University}

\date{}

\begin{document}

\maketitle

\begin{abstract}
We consider strategic games that are inspired by Schelling's model
of residential segregation. In our model, the agents are partitioned into $k$
types and need to select locations on an undirected graph. Agents can be either stubborn, 
in which case they will always choose their preferred location, or strategic, 
in which case they aim to maximize the fraction of agents of their own type
in their neighborhood. We investigate the existence of equilibria in these games, 
study the complexity of finding an equilibrium outcome or an outcome with high social welfare, 
and also provide upper and lower bounds on the price of anarchy and stability.
Some of our results extend to the setting where the preferences of the agents over their
neighbors are defined by a social network rather than a partition into types.
\end{abstract}

\section{Introduction} \label{sec:intro}
In 2015, African Americans constituted 83\% of the population of the City of Detroit. 
At the same time, the neighboring Oakland County was 77\% white, and in the city of Dearborn in Detroit metropolitan 
area about 30\% of the residents were Arab Americans. Similar phenomena can be observed
in many other major metropolitan areas around the world. In the developed world, the leading 
cause of such population patterns is not direct discrimination, which is typically illegal;
rather, it is the residents themselves who tend to select neighborhoods 
where their ethnic or social group is well-represented. \citet{S69,S71} 
proposed the following stylized model of this phenomenon: Agents of two different types are placed 
on a line or on a grid, and are assumed to be happy if at least a fraction $\tau$ of the  
agents within distance $w$ from them are of the same type, for some parameters $\tau$ and $w$; unhappy 
agents can either jump to empty positions or swap positions with other agents. Using simple 
experiments, Schelling showed that, even in cases where the agents are not opposed to integration 
($\tau < 1/2$), this behavior leads to almost complete segregation.

In the 50 years since Schelling's pioneering paper, this segregation model attracted the attention 
of many researchers, mostly in sociology and economics~\citep{AL93,BW07,BHO09,CF08,PV07,Y01,Z04a,Z04b}, 
but recently also in computer science~\citep{BEL14,BEL15,BIKK12,IKLZ17}. 
While the early work in this area was mainly empirical, 
the more recent papers have provided theoretical analysis. 
In particular, it was proved that the local behavior of unhappy agents 
is likely to create very large regions consisting of agents of the same type, even when $\tau$ is small, 
i.e., even when the agents themselves are tolerant towards having neighbors of the other type. 
The vast majority of this work was based on Schelling's original model, where agents' behavior
was explained by a simple stochastic model rather than strategic considerations. 

An alternative approach is to assume that the behavior of each agent is {\em strategic}, 
and exploit tools and techniques from non-cooperative game theory 
to analyze the induced games. 
To the best of our knowledge, there are only two papers in the literature that pursue this agenda.
Specifically, \citet{Z04b} considered a model with transferable utility where agents prefer 
to be in a balanced neighborhood. More recently, \citet{CLM18} investigated a setting 
that is closer to Schelling's motivating scenario, and also incorporates the idea
that, in addition to preferences over the composition of their neighborhood, 
agents may also have preferences over locations. In the model of \citet{CLM18}, 
there are two types of agents, and an agent $i$'s {\em happiness ratio}
is defined as the fraction of agents of $i$'s type among $i$'s neighbors.
Each agent has two further parameters: a tolerance threshold $\tau\in(0, 1)$ and a preferred location. 
An agent's primary goal is to find a location where her happiness ratio exceeds 
the tolerance threshold; if no such location is available, she aims to maximize
her happiness ratio. An agent's secondary goal is to minimize the distance to her preferred location.
To achieve these goals, agents can either swap locations (swap games) or jump to unoccupied locations
(jump games). The main contribution of the paper is to identify conditions under which agents
are guaranteed to converge to an equilibrium; for instance, the authors establish
that in jump games, convergence is guaranteed if agents have no preferred
locations and the underlying network is a ring.

\subsection{Our contribution}
The model of \citet{CLM18} makes an important contribution to the literature by enriching
Schelling's model with two additional components: agents who are fully strategic, 
and location preferences. However, the resulting model of agents' preferences is quite complex, 
and, consequently, not easy to analyze: the positive results in the paper
are limited to special cases of the utility function and highly regular networks.
In this paper, we propose a simpler model that aims to capture the same phenomena 
and is more amenable to formal analysis.

Specifically, just as in the work of \citet{CLM18}, 
in our basic model the agents are partitioned into $k$ types
and the set of available locations is represented by an undirected graph, 
which we will refer to as the {\em topology}.
We also incorporate location preferences in our model; however, instead of assuming 
that optimizing the distance to the preferred location is the secondary goal of every agent, 
we assume that agents are either {\em stubborn}, in which case they stay at their chosen 
location irrespective of their surroundings, or {\em strategic}, in which case they
aim to maximize their happiness ratio by jumping to an unoccupied location (we do not consider
swaps in this paper). Our model captures the fact that, in practice, many residents are unwilling 
to move to another area even if they are no longer satisfied with the composition of their neighborhood.
Importantly, unlike \citet{CLM18} or Schelling in his original work, we do not assume that agents have tolerance
thresholds; rather, a strategic agent is willing to move as long as there exists another
location with a better happiness ratio. Towards the end of the paper (Section~\ref{sec:extensions}), 
we also discuss several variants of this basic model. In particular, we show that some of our positive results
extend to the setting where there are no types, but rather the agents are connected by a social network
and care about the fraction of their friends (i.e., their neighbors in the social network)
among their neighbors in the topology; we refer to the resulting class of games as 
{\em social Schelling games}.

The rest of the paper is organized as follows.
We define our model in Section~\ref{sec:prelim}. Then, in Section~\ref{sec:existence}, 
we show that for some classes of topologies, such as stars and graphs of maximum degree two, 
our games always admit a pure Nash equilibrium, i.e., the strategic agents can be assigned
to the nodes of the topology so that none of them wants to move to a different location;
this result holds even for social Schelling games. In contrast, an equilibrium may fail 
to exist even if the topology is acyclic and has maximum degree four.
In Section~\ref{sec:complexity}, we complement this result by 
presenting a dynamic programming algorithm that decides whether an equilibrium 
exists on a tree topology; this algorithm runs in polynomial time if the number of types
is bounded by a constant. For more general topologies, we prove that deciding whether an equilibrium 
exists is an NP-complete problem. Similar hardness and easiness results hold for the problem
of maximizing the social welfare (the total utility of all strategic agents).
In Section~\ref{sec:poa-pos}, we study the effect of the strategic behavior  on 
the social welfare, by bounding the price of anarchy~\citep{KP99} and the price of stability~\citep{ADKTWR08}. 
In particular, we show that even in the absence of stubborn agents it may be
impossible to achieve the maximum social welfare in equilibrium. 
In Section~\ref{sec:extensions} we discuss several variants and extensions of our model
and establish some preliminary results for these new models, as well as outline directions
for future work.

\subsection{Other related work}
For an accessible introduction to the Schelling model and a the survey of the literature on non-strategic
variants, see chapter~4 in the book of \citet{EK10}, and the papers by \citet{BIKK12} and \citet{IKLZ17}.

Besides the work of \citet{CLM18}, which was discussed in detail earlier,   
our model shares a number of properties with hedonic games
\citep{DG80,BJ02}; these are games where agents split into coalitions, 
and each agent's utility is determined by the composition of her coalition. 
Specifically, in {\em fractional hedonic games} \citep{ABH14} the relationships among the agents 
are described by a weighted directed graph, where the weight of an edge $(i, j)$
is the value that agent $i$ assigns to agent $j$, and an agent's utility 
for a coalition is her average value for the other members in the coalition. If the graph
is undirected and all edge weights take values in $\{0, 1\}$, it can be interpreted 
as a friendship relation; then an agent's utility in a coalition is computed 
as the fraction of her friends among the coalition members, which is very similar 
to how utilities are defined in social Schelling games. On the other hand, the type-based
model is closely related to the Bakers and Millers game discussed by \citet{ABH14}.
This connection between Schelling games and hedonic games 
motivates much of the discussion in Section~\ref{sec:extensions}.
Of course, a fundamental difference between hedonic games and our setting is that
in the former agents derive their utilities from pairwise disjoint coalitions, whereas
in our model utilities are derived from (overlapping) neighborhoods. 

\section{The Model} \label{sec:prelim}
Let $N=\{1, \dots, n\}$ be a set of $n \geq 2$ {\em agents}.
The agents are partitioned
into $k \geq 2$ different {\em types} $T_1, \dots, T_k$
so that $\cup_{j=1, \dots, k} T_j = N$; we write ${\cal T}=(T_1, \dots, T_k)$.
We say that two agents $i, j\in N$, $i\neq j$, are {\em friends} if $i, j\in T_\ell$
for some $\ell\in [k]$; otherwise we say that $i$ and $j$ are {\em enemies}.
For each $i\in N$, we denote the set of all friends of agent $i$ by $F(i)$.

A {\em topology} is an undirected graph $G=(V,E)$ with no self-loops.
Each agent in $N$ has to select a node of this graph so that there are no collisions.
The agents are classified as either {\em strategic} or {\em stubborn};
let $R$ and $S$ denote these sets of agents so that $R \cup S = N$.
Stubborn agents care about their location only: each stubborn agent
has a preferred node and never moves away from that node. Thus, the preferences
of stubborn agents can be described by an injective mapping $\lambda:S\to V$;
for each $i\in S$ the node $\lambda(i)$ is the preferred node of agent $i$.
In contrast, strategic agents do not care about their location, but want
to be in a neighborhood that has a large proportion of their friends,
and are willing to move to a currently unoccupied node in order to increase
their utility.

Formally, given a set of agents $N=R\cup S$ with $|N|=n$, a topology $G=(V, E)$ with $|V|> n$
and a mapping $\lambda:S\to V$, an {\em assignment} is a vector
$\vv=(v_1,\dots, v_n)\in V^n$ such that (1) $v_i=\lambda(i)$ for each $i\in S$
and (2) $v_i\neq v_j$ for all $i, j\in N$ such that $i\neq j$;
here, $v_i$ is the node of the topology where agent $i$ is positioned.
A node $v \in V$ is {\em occupied} by agent $i$ if $v = v_i$.
For a given assignment $\vv$ and an agent $i\in N$,
let $N_i(\vv)=\{j \in N: \{v_i,v_j\} \in E \}$ be the set of neighbors of
agent $i$. Let $f_i(\vv)=|N_i(\vv) \cap F(i)|$ be the number of neighbors of $i$
in $\vv$ who are her friends. Similarly,
let $e_i(\vv) = |N_i(\vv)| - f_i(\vv)$ be the number
of neighbors of $i$ in $\vv$ who are her enemies.
Following \citet{CLM18}, 
we define the utility $u_i(\vv)$ of an agent $i\in R$ in $\vv$ to be $0$ if $f_i(\vv)=0$; otherwise, her utility is defined as the fraction of her friends among the agents in the neighborhood:
$$
u_i(\vv) = \frac{f_i(\vv)}{f_i(\vv)+e_i(\vv)}.
$$

A tuple $I=(R,S,{\cal T},G, \lambda)$, where $R$ is the set of strategic agents,
$S$ is the set of stubborn agents, ${\cal T}=(T_1, \dots, T_k)$ is a list of types,
$G=(V, E)$ is a topology that satisfies $|V|> |R|+|S|$,
and $\lambda$ is an injective mapping from $S$ to $V$,
is called a {\em $k$-typed Schelling game} or {\em $k$-typed instance};
let $\calI$ be the set of all possible games.
We say that an assignment $\vv$ is a {\em pure Nash equilibrium}
(or, simply, {\em equilibrium}) of $I$ if
no strategic agent $i$ has an incentive to unilaterally deviate to an empty node $z$ of $G$
in order to increase her utility, i.e.,
for every $i\in R$ and for every node $z\in V$ such that $z\neq v_j$
for all $j\in R\cup S$ it holds that
$u_i(\vv) \geq u_i(z,\vv_{-i})$,
where $(z,\vv_{-i})$ is the assignment obtained by changing the $i$-th entry of
$\vv$ to $z$.
Let $\EQ(I)$ denote the set of all equilibria of game $I$.

The {\em social welfare} of an assignment $\vv$ is defined as the total utility
of all strategic agents:
$$
\SW(\vv) = \sum_{i \in R} u_i(\vv).
$$
Let $\vv^*(I)$ be an assignment that maximizes the social welfare for a given game $I$; we refer to it as an {\em optimal} assignment.

The {\em price of anarchy} (PoA) of game $I$ with at least one equilibrium is the ratio between the 
optimal social welfare and the social welfare of the {\em worst} equilibrium; its 
{\em price of stability} (PoS) is defined as the ratio between the optimal 
social welfare and the social welfare of the {\em best} equilibrium:
\begin{align*}
\PoA(I) &= \sup_{\vv \in \EQ(I)} \frac{\SW(\vv^*(I))}{\SW(\vv)}, \\
\PoS(I) &= \inf_{\vv \in \EQ(I)} \frac{\SW(\vv^*(I))}{\SW(\vv)}.
\end{align*}
The price of anarchy and the price of stability are the suprema of $\PoA(I)$ and $\PoS(I)$ over all
$I \in \calI$ such that $\EQ(I)\neq \emptyset$, respectively.

\section{Existence of Equilibria} \label{sec:existence}
In this section, we focus on the existence of equilibria.
We warm up by observing that for highly structured topologies such as
paths, rings, and stars, there is always at least one equilibrium assignment,
and some such assignment can be computed efficiently. This can be shown directly, 
and also follows from a more general result established in Section~\ref{sec:extensions}
(Theorem~\ref{thm:social-star-deg2}).

\begin{theorem}\label{thm:typed-star-deg2}
Every $k$-typed Schelling game
where the topology is a star or a graph of maximum degree $2$
admits at least one equilibrium assignment, which can be computed in polynomial time.
\end{theorem}

However, in general, an equilibrium may fail to exist;
this holds even if the topology is acyclic and there are no stubborn agents.

\begin{figure}[t]
\center
\includegraphics[scale=0.45]{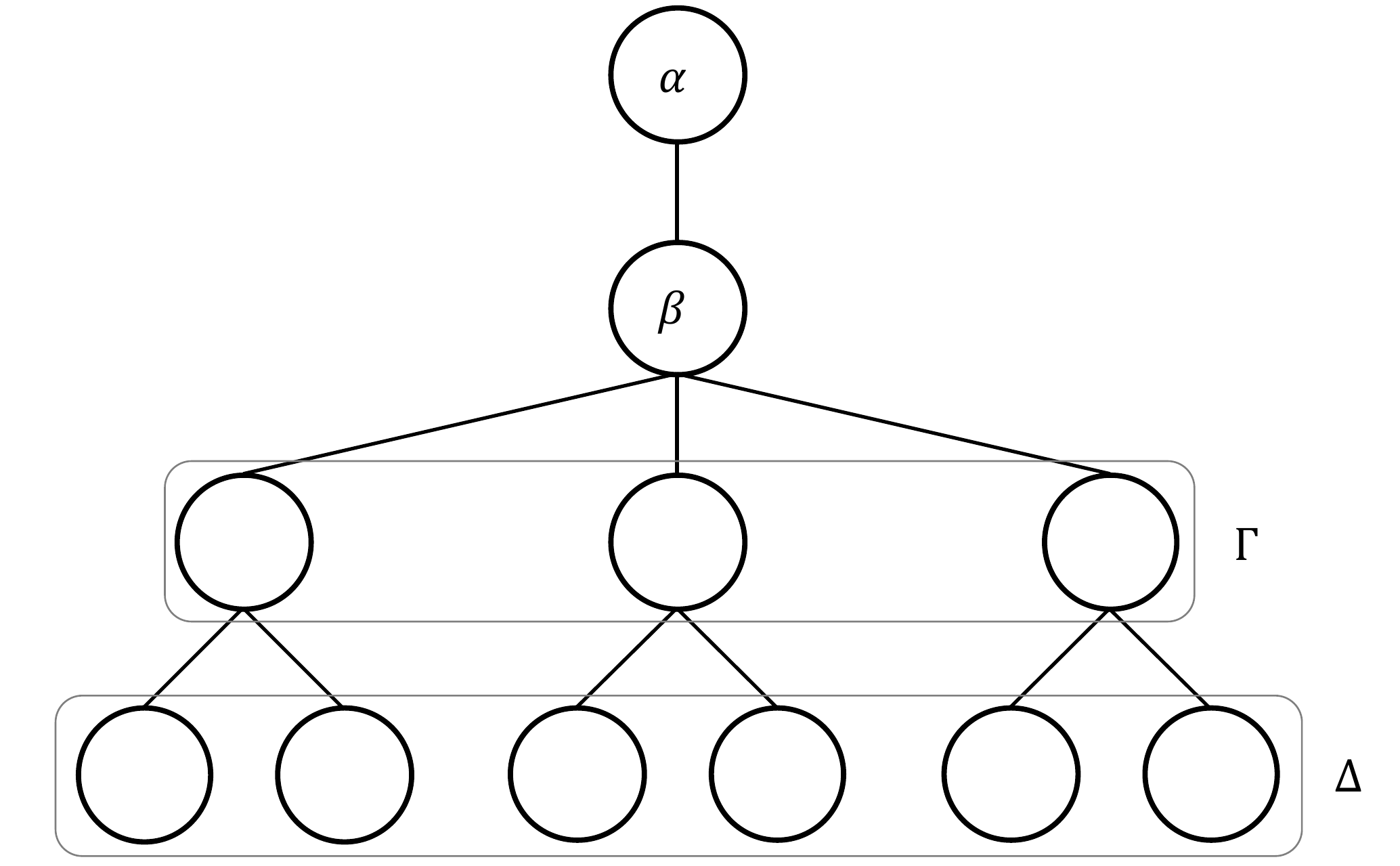}
\caption{Example of the topology used in the proof of Theorem~\ref{thm:non-existence} for $k=2$.}
\label{fig:non-existence}
\end{figure}

\begin{theorem}\label{thm:non-existence}
For every $k\ge 2$ there exists a $k$-typed instance $(R, S, {\cal T}, G, \lambda)$
where $S=\emptyset$ and $G$ is a tree that does not admit an equilibrium.
\end{theorem}

\begin{proof}
Given $k\ge 2$, we construct an instance with $2k+1$ agents per type; the total number of agents
is $n=k(2k+1)$. The topology $G=(V,E)$ is a tree that consists of $|V|=n+1$ nodes, which
are distributed over four layers. Specifically, the tree has a root $\alpha$,
which has one child $\beta$. Node $\beta$ has $2k-1$ children; we denote the set of its children by $\Gamma$.
Each node in $\Gamma$ has $k$ children, which are the leaves of the tree;
we denote the set of all leaves by $\Delta$.
Figure~\ref{fig:non-existence} depicts the topology for $k=2$. Now, assume that there is an
equilibrium assignment; note that exactly one node is left empty.
We consider four cases depending on the location of the empty node.
\begin{description}
\item {\bf Node $\alpha$ is empty.} Assume that the agent occupying node $\beta$ is of type $T$.
Then, since there are $2k$ other agents of type $T$ and there are only $2k-1$ nodes in $\Gamma$, 
there must exist some subtree rooted
at a node in $\Gamma$ that contains both agents of type $T$ and agents that belong to other types.
Then an agent of type $T$ from this subtree has an incentive to deviate to $\alpha$.

\item {\bf Node $\beta$ is empty.} Assume that the agent occupying node $\alpha$ is of type $T$;
note that her utility is $0$.
If she does not have an incentive to deviate to $\beta$, it follows that
no agent of type $T$ occupies a node in $\Gamma$. But then there is an agent of type $T$
who occupies a node in $\Delta$; as her parent is not of type $T$, her utility is $0$,
and she can increase it by moving to $\beta$.

\item {\bf Some node $\gamma \in \Gamma$ is empty.} Consider the agents occupying the children of $\gamma$;
note that their utility is $0$. If at least two of them have the same type, each of them
has an incentive to deviate to $\gamma$ in order to increase her utility to at
least $\frac{1}{k}$. If all of them have different types, then there is
exactly one agent of each type in this set. In particular, there is an agent $i$ who has the
same type as the agent occupying $\beta$; then $i$ can move to $\gamma$ to increase her utility.

\item
{\bf Some node $\delta \in \Delta$ is empty.} Let $\gamma$ denote the parent of this node,
and suppose that $\gamma$ is occupied by an agent $i$ of type $T$. We say that
an agent $j$ of type $T$ is {\em hungry} if $j\neq i$ and $j$ is adjacent to at least one agent
of a different type; note that a hungry agent has an incentive to deviate to $\delta$.
We claim that at least one agent is hungry.
Indeed, if $\beta$ is occupied by an agent $j$ of type $T$, then either $j$ is hungry
or every agent in $\Gamma\setminus\{\gamma\}$ is hungry.
If the agent in $\beta$ is not of type $T$ and there is an agent $\ell$ of type $T$ in
$\Gamma\setminus\{\gamma\}$, then $\ell$ is hungry. Finally, if no agent in
$\Gamma\setminus\{\gamma\}$ is of type $T$, there exists a leaf node
not in $\gamma$'s subtree that is occupied by an agent $r$ of type $T$;
$r$ is then hungry.
\end{description}
The proof is complete.
\end{proof}


\section{Computational Complexity}\label{sec:complexity}
We now turn our attention to the computational complexity of
$k$-typed Schelling games. The main result of this section is that 
finding an equilibrium assignment is computationally intractable.

\begin{theorem}\label{thm:eq-hardness}
For every $k\ge 2$, given a $k$-typed Schelling game $I$,
it is {\em NP}-complete to decide whether $I$ admits an equilibrium assignment.
The hardness result holds even if all strategic agents belong to the same type.
\end{theorem}

\begin{proof}
We give a proof for $k=2$; it is straightforward to extend it to $k\ge 2$.
We will use a reduction from the {\sc Clique} problem.
An instance of this problem is an undirected graph $H=(X, Y)$ and an integer $s$;
it is a yes-instance if $H$ has a complete subgraph of size $s$.
Given an instance $\langle H, s\rangle$ of {\sc Clique}
with $H=(X, Y)$, we assume without loss of generality that $s \ge 5$
and construct an instance of our problem as follows:
\begin{itemize}
\item There are two agent types: red and blue.
\item There are $s$ strategic red agents; all remaining agents are stubborn.
We will describe the stubborn agents and their locations when defining the topology.
\item The topology $G=(V,E)$ consists of three disjoint components
$G_1$, $G_2$, and $G_3$ such that
\begin{itemize}
\item
$G_1=(V_1, E_1)$, where
$V_1=X \cup W$, $|W|=s-2$, $E_1=Y \cup \{\{v,w\}: v \in X, w\in W\}$.
There is a stubborn blue agent 
at each node $w\in W$;
\item
$G_2$ is a complete bipartite graph with parts $L$ and $R$,
$|L|=s-2$, $|R|=4s$. Of the $4s$ nodes in $R$,
$2s+1$ nodes are occupied by red agents and 
$2s-1$ nodes are occupied by blue agents;
\item
$G_3$ has three empty nodes, denoted $x$, $y$, and $z$, 
and $121$ nodes --- $41$ red and $80$ blue --- occupied by stubborn agents.
There is an edge between nodes $x$ and $y$; also,
$x$ is connected to $1$ red agent and $2$ blue agents;
$y$ is connected to $41$ red agents and $80$ blue agents, and
$z$ is connected to $5$ red agents and $7$ blue agents.
\end{itemize}
\end{itemize}

Note that a strategic red agent obtains a utility of 
$\frac{2s+1}{4s} = \frac{1}{2} + \frac{1}{4s}$
by choosing an available node in $G_2$ and a utility of
$\frac{5}{12}$ by choosing $z$.
If she chooses $x$, her utility is
$\frac{1}{3}$ if $y$ is unoccupied and
$\frac{1}{2}$ otherwise.
Similarly, if she chooses $y$, her utility is
$\frac{41}{121}$ if $x$ is unoccupied and
$\frac{42}{122}$ otherwise;
note that $\frac{1}{3} < \frac{41}{121} < \frac{42}{122}< \frac{5}{12}$.

Now, suppose that $G$ contains a clique of size $s$. If strategic red agents
occupy the nodes of that clique, the utility of each such agent is 
$\frac{s-1}{(s-1)+(s-2)} = \frac{1}{2} + \frac{1}{4s-6}$.
Thus, by our choice of parameters, no agent has a profitable deviation.

On the other hand, suppose that $G$ does not contain a clique of size $s$.
Assume for the sake of contradiction that there is an equilibrium assignment $\vv$.

Suppose first that in $\vv$ some strategic agents
are located in $G_1$. It cannot be the case that each of them 
is adjacent to $s-1$ friends, as this would mean that their locations 
form a clique of size $s$. Hence, at least one of these agents is adjacent 
to at most $s-2$ friends. As this agent is also adjacent to the $s-2$ stubborn blue agents 
in $W$, her utility is at most $\frac{1}{2}$.
By our choice of parameters, all unoccupied nodes of $G_2$ offer a higher utility, namely,  
$\frac{1}{2} + \frac{1}{4s}$. Thus, if there are strategic agents in $G_1$, 
all $s-2$ nodes of $G_2$ that are available to strategic agents 
must be occupied. But then, there are at most two strategic agents
in $G_1$, which means that their utility 
is at most $\frac{1}{s-1} < \frac{1}{3}$ (recall that we assume that $s\ge 5$).
This leads to a contradiction, as these strategic agents would be 
better off moving to $G_3$ where their utility would be at least $\frac{1}{3}$.

Therefore, in equilibrium no strategic agent can be located at a node of $G_1$.
Further, since all unoccupied nodes of $G_2$ always offer 
more utility than any unoccupied nodes of $G_3$ can offer,
in equilibrium all nodes of $G_2$ are occupied, and  
the two remaining strategic agents must be in $G_3$, 
with one of $x$, $y$, and $z$ left empty.

Suppose that $z$ is empty. Then
the agent located at $y$ can increase
her utility from $\frac{42}{122}$ to $\frac{5}{12}$ by moving to $z$, a contradiction.
If $y$ is empty, the agent located at $x$ can increase
her utility from $\frac{1}{3}$ to $\frac{41}{121}$ by moving to $y$, a contradiction.
Finally, if $x$ is empty, the agent located at $z$ can increase
her utility from $\frac{5}{12}$ to $\frac{1}{2}$ by moving to $x$, a contradiction.
As we have exhausted all possibilities, it follows that if $G$ does not have a clique
of size $s$, then there is no equilibrium assignment.
\end{proof}

The proof of Theorem~\ref{thm:eq-hardness} can be adapted to show 
that maximizing social welfare in Schelling games is NP-hard as well.

\begin{theorem}\label{thm:opt-hardness}
For every $k\ge 2$, given a $k$-typed Schelling game $I$ and a rational value $s$,
it is {\em NP}-complete to decide whether $I$ admits an assignment
with social welfare at least $s$.
The hardness result holds even if $k=2$, 
all strategic agents belong to one type, and the other type consists
of a single stubborn agent.
\end{theorem}

\begin{proof}
We modify the reduction in the proof of Theorem~\ref{thm:eq-hardness}
by removing the gadgets $G_2$ and $G_3$ and replacing the set $W$
with a single node $w$.
That is, given an instance $\langle H, s\rangle$ of {\sc Clique},
we construct an instance of our social welfare maximization problem as follows:
\begin{itemize}
\item There are two agent types: red and blue.
\item There are $s$ strategic red agents and one stubborn blue agent.
\item The topology $G=(V,E)$ is defined so that $V=X \cup \{w\}$ and $E=Y \cup \{\{v,w\}: v \in X\}$.
\item The single stubborn blue agent is positioned at node $w$.
\end{itemize}
Note that the utility of a red agent $p$ in an assignment $\vv$ is $\frac{r}{r+1}$,
where $r$ is the number of red agents that $p$ is adjacent to in $\vv$; the function
$\frac{r}{r+1}$ is increasing in $r$ and we have $r\le s-1$ for any assignment.
Hence, the social welfare of $s-1$ can be achieved if and only if
the red agents can be placed in $G$ so that each agent
is adjacent to every other red agent, in which case the utility of
each strategic agent is $\frac{s-1}{s}$;
this is possible if and only if $H$ contains a clique of size $s$.
\end{proof}

On the positive side, for small $k$ we can efficiently decide whether an equilibrium exists
if the topology $G$ is a tree. Our algorithm is based on
dynamic programming: it selects an arbitrary node of $G$ to be the root, 
and then for every node $v$ of $G$,
it fills out a multidimensional table
whose dimension is linear in the number of types, proceeding from the leaves to the root.
It decides whether the given instance admits
an equilibrium by scanning the table at the root node.
The details of the algorithm are given in the appendix. 

\begin{theorem}\label{thm:existence-tree}
Given a $k$-typed Schelling game $I$ with $n$ agents, where the topology $G$ is a tree,
we can decide whether $I$ admits an equilibrium (and compute one if it exists)
in time $\mathit{poly}(n^k)$,
i.e., this problem lies in the complexity class {\em XP} 
with respect to the number of types $k$.
\end{theorem}

By slightly modifying our algorithm, we can compute
an assignment that maximizes the social welfare, 
either among all assignments or among equilibria.

\begin{corollary}\label{thm:optimal-tree}
Given a $k$-typed Schelling game, where $G$ is a tree,
the problems of computing an equilibrium with maximum social welfare
or a socially optimal assignment are in {\em XP} with respect to $k$.
\end{corollary}

We have not been able to determine whether the problem of computing an equilibrium assignment
is fixed-parameter tractable with respect to the number of types;
we leave this question for future work.


\section{Price of Anarchy and Stability}\label{sec:poa-pos}
In this section, we investigate the loss in social welfare caused by strategic behavior, 
as measured by the price of anarchy and the price of stability; unless otherwise specified, the topology is assumed to be a connected graph.\footnote{We can easily observe that the price of anarchy can be unbounded for not connected topologies. For instance, consider a topology with one isolated node and a connected subgraph $W \subset V$ such that $|W|=n$. Then, any assignment of the $n$ agents at the nodes of $W$ is an equilibrium. Hence, there might exist equilibria with zero social welfare.}
We start by establishing bounds on the price of anarchy for instances with no stubborn agents.

\begin{theorem}\label{thm:poa-strategic}
For $k$-typed Schelling games
with no stubborn agents and $n$ strategic agents,
the PoA
\begin{itemize}
\item can be unbounded for each $k\ge 2$;
\item is $\Theta(n)$ when there are at least two agents per type;
\item is $k+o(1)$ if each type has the same number of agents.
\end{itemize}
\end{theorem}

We prove each statement separately.

\begin{lemma}
For $k$-typed Schelling games
with no stubborn agents and $n$ strategic agents,
the price of anarchy can be unbounded for each $k\ge 2$.
\end{lemma}

\begin{proof}
Consider an instance with two agents of type $T_1$ and one agent of every other type $T_i$, $i \geq 2$.
The topology is a star with $k+2$ nodes.
Then, any assignment where the center node is occupied by an agent of type $T_i$ with $i \geq 2$,
is an equilibrium with zero social welfare. In contrast,
any assignment where the center node is occupied by an agent of type $T_1$ is again an
equilibrium, but the social welfare is now strictly positive (since the two agents
of type $T_1$ are connected to each other).
\end{proof}

\begin{lemma}
For $k$-typed Schelling games
with no stubborn agents and $n$ strategic agents,
in which there are at least two agents
per type, the price of anarchy is $\Theta(n)$ for every $k\ge 2$. 
\end{lemma}

\begin{proof}
For the lower bound, consider an instance with $|T_i|=2$ for $i=1, \dots, k-1$
and $|T_k|=n-2(k-1)$.
The topology is a star with $n+1$ nodes.
Then any assignment where an agent that belongs to one of the first $k-1$ types 
is at the center of the star is an
equilibrium with social welfare $1+1/(n-1) \leq 2$, while for any assignment
where an agent of type $T_k$ is at the center of the star
the social welfare is $n-2(k-1)-1+1/(n-1) \geq n-2k$.
Hence, the price of anarchy is at least $n/2 - k$.

For the upper bound, consider a $k$-typed instance
with $n_i \geq 2$ agents of each type $T_i$,
so that $n = \sum_{i \in [k]} n_i$. We will show that
the social welfare of any equilibrium assignment is at least $1$.
This implies our bound on the price of anarchy, since the optimal social welfare is at most $n$.

Let $\vv$ be an arbitrary equilibrium assignment.
Recall that we assume that the number of available nodes exceeds the number of agents and the topology is connected,
so there must exist some empty node $v$ with at least one non-empty neighbor. Suppose that
$v$ is connected to $x_i$ agents of type $T_i$, for $i \in [k]$, and let $s = \sum_{i \in [k]}x_i$.
By deviating to $v$, an agent of type $T_i$
would get utility $\frac{x_i}{s}$
if she is not connected to $v$, and utility $\frac{x_i-1}{s - 1}$ otherwise; again, for readability, we use the convention that $\frac00=0$.
Since at equilibrium no agent has an incentive to deviate,
her utility is at least the utility she would get by deviating to $v$. Therefore,
the social welfare at equilibrium is at least
\begin{align*}
\SW(\vv) &\geq \sum_{i \in [k]} \left( (n_i - x_i)\frac{x_i}{s} +
x_i \frac{x_i-1}{s - 1}  \right) \\
&\geq \frac{1}{s}\sum_{i \in [k]} (n_i-1)x_i \geq 1,
\end{align*}
where the last inequality holds since $n_i \geq 2$ for every $i \in [k]$.
This completes the proof.
\end{proof}

\begin{figure}[t]
\center
\includegraphics[scale=0.45]{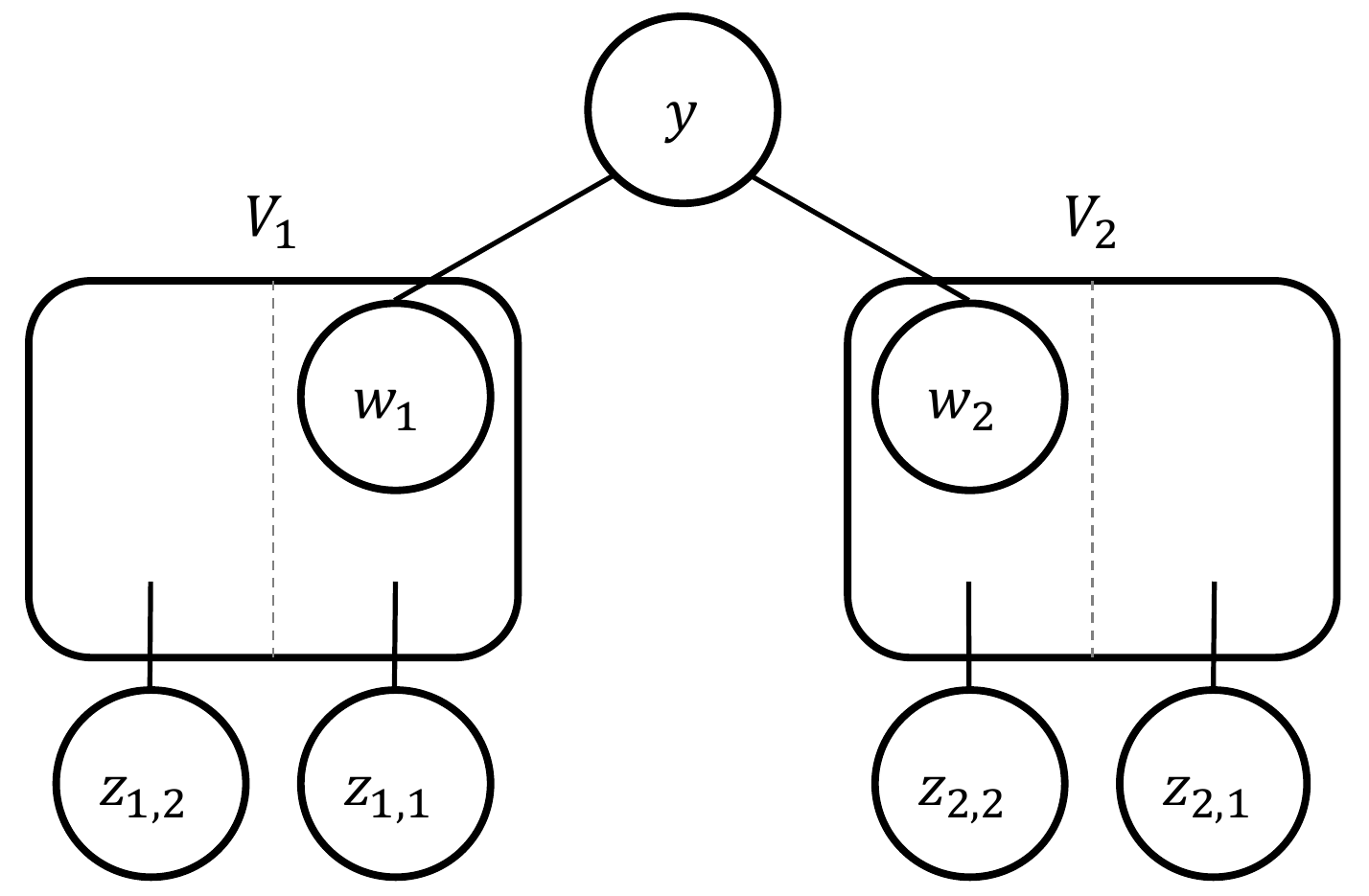}
\caption{The topology used in the proof of Theorem~\ref{lem:poa-strategic} for $k=2$.}
\label{fig:poa-strategic}
\end{figure}

\begin{lemma}\label{lem:poa-strategic}
For $k$-typed Schelling games
with no stubborn agents and $n$ strategic agents,
the PoA
is $k+o(1)$ if each type has the same number of agents.
\end{lemma}

\begin{proof}
For the lower bound, fix $k\ge 2$,
let $\ell \geq 2$ be a parameter and consider an instance with $k(\ell+1)$ agents per type;
altogether there are $n=k^2(\ell+1)$ agents.
The topology consists of $n+1$ nodes and is defined as follows.
There are $k$ cliques $V_1, \dots, V_k$ of size $k\ell$ each, and a node $y$. In each clique
$V_i$ there is a special node $w_i$ that is connected to $y$.
Also, for each $i\in [k]$ there are $k$ auxiliary nodes
$z_{i,1}, \dots, z_{i,k}$; each of these nodes is connected to a distinct set of $\ell$
nodes in $V_i$. Let $z_{i,i}$ be the auxiliary node that is connected to $w_i$.
Figure~\ref{fig:poa-strategic} illustrates this topology for $k=2$.

There is an optimal assignment where
all $k(\ell+1)$ agents of type $T_i$ are placed at the nodes of clique
$V_i$ and the corresponding auxiliary nodes, so that all agents are connected only to agents of the
same type and have maximum utility $1$. Therefore,
the optimal social welfare is $k^2(\ell+1)$.

In contrast, consider the following equilibrium assignment: node $y$ is empty,
and for each $i, j\in [k]$ all $\ell$ nodes in
$V_i$ that are connected to the auxiliary node $z_{i,j}$ as well as $z_{i,j}$ itself are occupied by
agents of type $T_j$. Since node $y$ is connected to $k$ nodes that are occupied by
agents of different types, any agent would get utility $1/k$ by deviating there. No agent occupying
an auxiliary node has an incentive to deviate since she is connected only to agents of her type. 
For every clique, each agent is connected to exactly $\ell$ agents of the same
type ($\ell-1$ of whom occupy nodes of the clique and one that occupies the corresponding auxiliary node)
and $(k-1)\ell$ agents of different type; thus, her utility is $1/k$. Consequently, no agent has
an incentive to deviate, and the social welfare 
is $k\cdot k\ell \cdot \frac{1}{k} + k^2 = k(\ell + k)$.
Hence, the PoA is at least
$\frac{k\ell+k}{\ell+k}$; this expression
becomes arbitrarily close to $k$ as $\ell$ grows.

For the upper bound, consider an arbitrary instance with $n$ agents and $k\geq 2$
types so that there are $n/k$ agents per type. We will show that the social welfare of any
equilibrium assignment is at least $n/k-1$. The bound on the PoA then follows,
since the optimal social welfare is at most $n$.

Recall that we assume that the number of available nodes exceeds the number of agents and the topology is connected,
so there must exist some empty node $v$ with at least one non-empty neighbor. Suppose that
$v$ is connected to $x_i$ agents of type $T_i$, for $i \in [k]$, and let
$s=\sum_{i\in [k]}x_i$.
Consider an agent of type $T_i$.
A deviation to $v$ would give her utility $\frac{x_i}{s}$
if she is not connected to $v$, and utility $\frac{x_i-1}{s - 1}$ otherwise
(for readability we use the convention that $\frac00=0$).
Since at equilibrium no agent has any incentive to deviate,
her utility is at least the utility she would get by deviating to $v$. Therefore,
the social welfare at equilibrium is at least
\begin{align*}
\SW(\vv) &\geq \sum_{i\in [k]} \bigg( \left( \frac{n}{k}-x_i \right)
\frac{x_i}{s} + x_i \frac{x_i-1}{s - 1} \bigg) \\
&\geq \frac{1}{s}{\sum_{i \in [k]} \bigg( \left( \frac{n}{k}-x_i \right) x_i + x_i(x_i-1) \bigg)}
= \frac{n}{k}-1.
\end{align*}
The proof is complete.
\end{proof}

In the setting considered in Theorem~\ref{thm:poa-strategic}
the PoA improves significantly if we require each type 
to have the same number of agents. In the presence of stubborn agents,
to ensure that the price of anarchy does not depend on the number 
of agents $n$, we additionally require 
that this constraint holds both for strategic and for stubborn agents.

\begin{theorem}\label{thm:poa-stubborn}
For $k$-typed Schelling games with $n$ agents the PoA
\begin{itemize}
\item is $\Omega(n)$ for each $k\ge 2$ even if there is an equal number of agents per type;
\item is $k+o(1)$ if each type has the same number
of strategic agents and the same number of stubborn agents.
\end{itemize}
\end{theorem}

Again, we prove each statement separately.

\begin{lemma}
For $k$-typed Schelling games
with no stubborn agents and $n$ strategic agents,
the price of anarchy is $\Omega(n)$ for each $k\ge 2$, 
even if there is an equal number of agents per type.
\end{lemma}

\begin{proof}
Pick a positive integer $\ell$ and 
consider an instance with $n=k\ell$ agents such that there are $\ell$ strategic agents of type $T_1$, 
one strategic agent and $\ell-1$ stubborn agents of type $T_2$, 
and $\ell$ stubborn agents of type $T_i$ for each $i=3, \dots, k$.
The topology is a star with $n+1$ nodes, and all stubborn agents occupy leaf nodes.
Then, any assignment where the strategic agent of type $T_2$ 
occupies the center node is an equilibrium with social
welfare $\frac{\ell-1}{k\ell-1} < \frac{1}{k}$, while the social welfare of any assignment
where the center node is occupied by an agent of type $T_1$ is
$\ell-1+\frac{\ell-1}{k\ell-1} > \ell-1$.
Hence, the price of anarchy is at least $k(\ell-1)=n-k$.
\end{proof}

\begin{figure}[t]
\center
\includegraphics[scale=0.45]{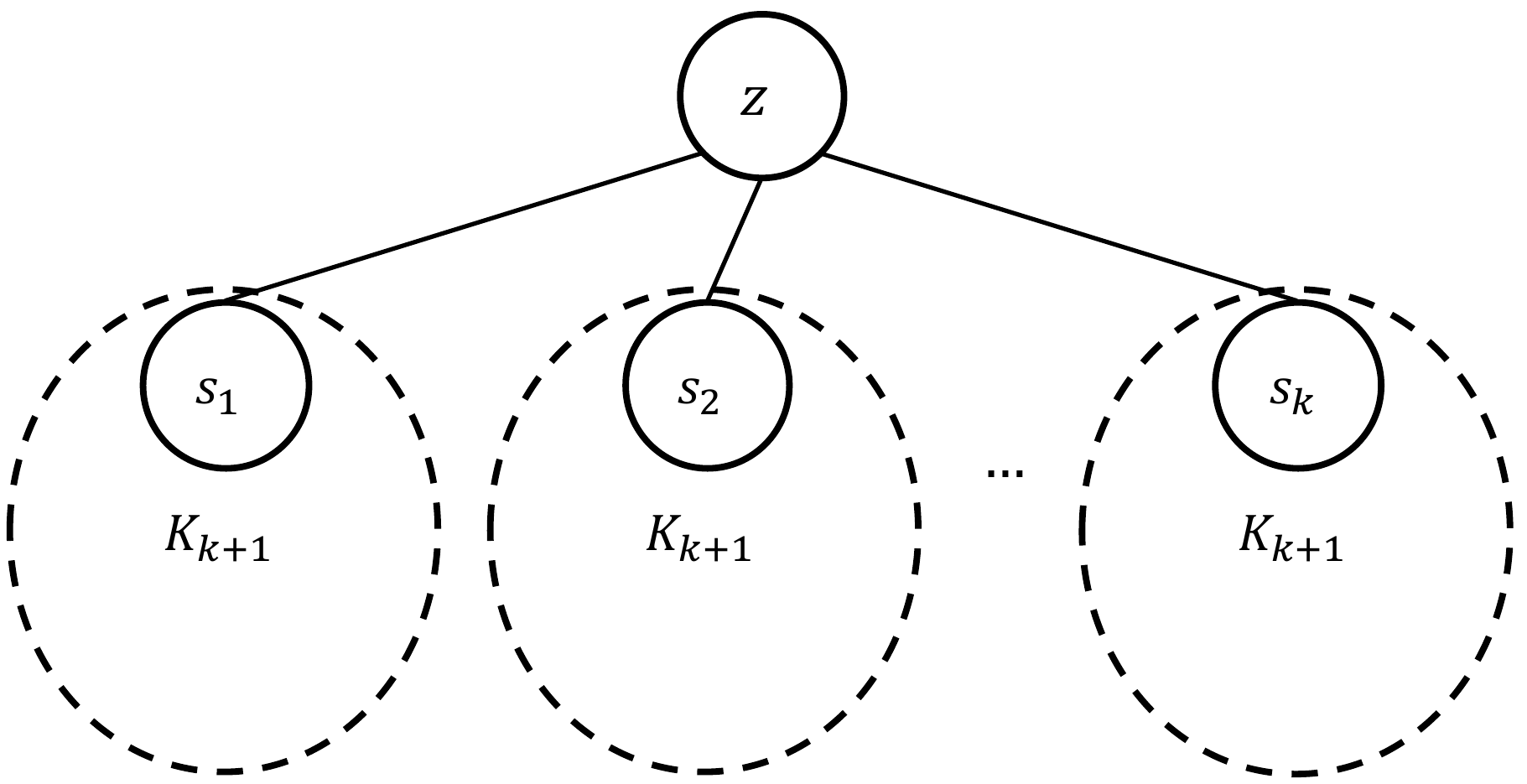}
\caption{The topology used in the proof of Lemma~\ref{lem:poa-k-equal-stubborn}.}
\label{fig:poa-stubborn}
\end{figure}

\begin{lemma}\label{lem:poa-k-equal-stubborn}
For $k$-typed instances,
where all types have the same number of strategic agents
                 and the same number of stubborn agents,
the price of anarchy is $k+o(1)$.
\end{lemma}

\begin{proof}
We first establish the lower bound. Suppose that $k$ is odd, 
and consider an instance with $k$ types of agents such
that there are $k$ strategic agents and one stubborn agent $s_i$ per type $T_i$.
The number of strategic agents is $s=k^2$.
The topology is depicted in Figure~\ref{fig:poa-stubborn} and consists
of $k$ cliques of size $k+1$, which are connected to each other via an auxiliary node $z$.
The stubborn agent $s_i$ of type $T_i$ occupies the node of the $i$-th clique
that is adjacent to $z$.

In an optimal assignment all strategic agents
of type $T_i$ occupy the nodes of the $i$-th clique: this ensures that the utility of each
strategic agent is $1$ and the social welfare is equal to $s$.
In contrast, consider an assignment
where the auxiliary node $z$ is left empty, and the $i$-th clique includes one agent of
type $T_i$ and $\frac{k-1}{2}$ pairs of agents of different types.\footnote{To see how such an
assignment can be computed, split $k-1$ agents of type $T_i$ into pairs and think
of each pair as one ball of color $T_i$ and each clique as a bin. 
Then, there are $\frac{k-1}{2}$ balls of each color, 
which must be placed in $k$ bins so that each bin contains $\frac{k-1}{2}$ balls of different color. 
To accomplish this, we can order the balls so that balls of type $T_i$ 
appear in positions $i, k+i, 2k+i, \dots $; hence, we can simply put the first $\frac{k-1}{2}$ 
balls in the first bin, the next $\frac{k-1}{2}$ balls in the second bin, etc.} 
This is an equilibrium since all
strategic agents have utility $1/k$, which is exactly the utility they would get by deviating to $z$.
Therefore, the social welfare achieved by this equilibrium assignment is $\frac{s}{k}$, and the price
of anarchy is at least $k$.

When $k$ is even, we can modify the instance as follows. For each $i\in [k]$, 
there are $k-1$ strategic agents and one stubborn agent per type $T_i$. 
The topology consists of $k-1$ cliques of size $k$, 
which are connected to each other via an auxiliary node $z$, 
together with $k$ dummy nodes each connected to a single node occupied by a stubborn player.
For $i \in [k-1]$, the stubborn agent of type $T_i$ 
occupies the node of the $i$-th clique that is adjacent to $z$, 
and the stubborn agent of type $T_k$ occupies one of the dummy nodes.

If all strategic agents of type $T_i$, for $i \in [k-1]$,
occupy the nodes of the $i$-th clique, and the agents of type $T_k$ occupy the dummy nodes, then
the social welfare is equal to $(k-1)^2$. 
On the other hand, there is an equilibrium where agents of type $T_k$ 
occupy dummy nodes and agents of other types are distributed over the cliques as in the equilibrium
for odd $k$. Then, for $i \in [k-1]$, the utility of each strategic agent of type $T_i$ 
is $\frac{1}{k-1}$. The social welfare in this case is $k-1$, and so the price of anarchy is at least $k-1$.

For the upper bound,
consider an arbitrary $k$-typed instance with $t$ strategic and $\ell$
stubborn agents per type, for some integers $t>0$ and $\ell \ge 0$.
We will show that the social welfare of any equilibrium assignment is at
least $t-1$. The bound then follows since
the utility of every strategic agent is at most $1$,
meaning that the optimal social welfare is at most $kt$.

Let $\vv$ be an arbitrary equilibrium assignment.
Since the number of available nodes exceeds the number of agents and the topology is connected,
there must exist some empty node $v$ with at least one non-empty neighbor. Suppose that
$v$ is connected to $x_i$ agents of type $T_i$, for $i \in [k]$,
and $x_i^R$ of them are strategic. Also, let $s = \sum_{i \in [k]}x_i$.
Now, consider a strategic agent of type $T_i$.
A deviation to $v$ would give her utility $\frac{x_i}{s}$
if she is not connected to $v$, and utility $\frac{x_i-1}{s - 1}$ otherwise; again, for readability, we use the convention $\frac00=0$.
Since at equilibrium no strategic agent has any incentive to deviate,
her utility is at least the utility she would get by deviating to $v$. Therefore,
the social welfare at equilibrium is at least
\begin{align*}
\SW(\vv)
&\geq \sum_{i \in [k]} \bigg( \left(t-x_i^R \right)
\frac{x_i}{s} + x_i^R \frac{x_i-1}{s - 1} \bigg) \\
&\geq \frac{1}{s}\sum_{i\in [k]} \bigg(t x_i - x_i^R \bigg)
\geq t-1,
\end{align*}
where the last inequality follows since $x_i^R \leq x_i$. The proof is complete.
\end{proof}

Finally, we show that in Schelling games
even the best equilibrium need not be socially optimal, 
even if all agents are strategic.\footnote{Note that the assumption of a connected topology is no longer necessary for meaningful bounds on the price of stability, since the PoS deals with the best-case equilibrium assignment rather than the worst-case one.}

\begin{theorem}\label{thm:pos}
For $k$-typed Schelling games the PoS
\begin{itemize}
\item can be unbounded for each $k\ge 2$;
\item is at least $3$ for each even $k\geq 2$, if there is the same number of stubborn agents
per type;
\item is at least $34/33$ for each $k\geq 2$, even in the absence of stubborn agents.
\end{itemize}
\end{theorem}

The proof of the above theorem follows by the next three lemmas.

\begin{figure}[t]
\center
\includegraphics[scale=0.45]{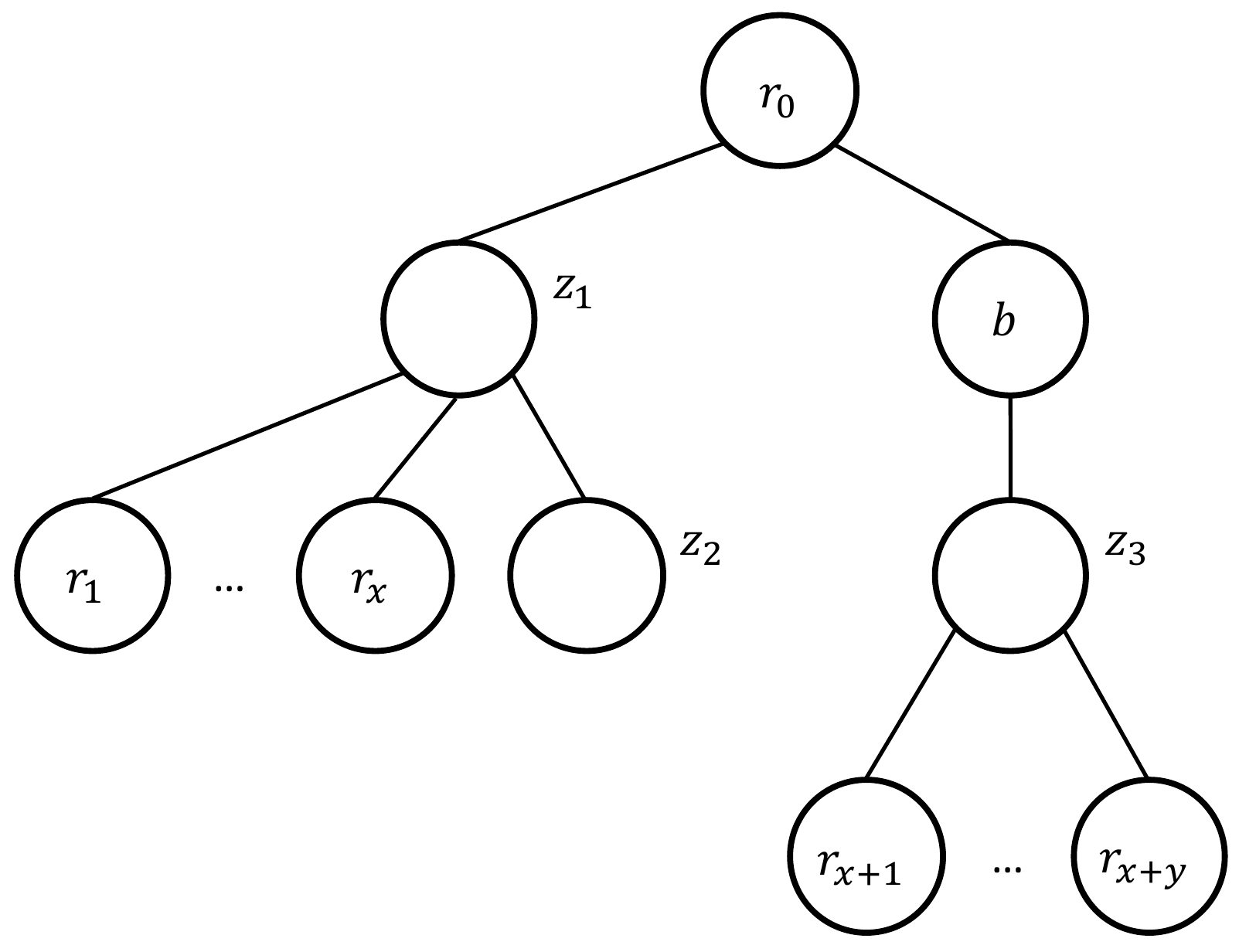}
\caption{The topology used in the proof of Lemma~\ref{lem:pos-stubborn-general} for $k=2$.}
\label{fig:pos-stubborn-general}
\end{figure}

\begin{lemma}\label{lem:pos-stubborn-general}
For $k$-typed Schelling games, the price of stability
can be unbounded for each $k \geq 2$.
\end{lemma}

\begin{proof}
We prove this lemma only for $k=2$; our instance can then be generalized to any number of types by adding isolated nodes in the topology which are occupied by stubborn players of different types. 

Let $\varepsilon > 0$ be a parameter such that $x=\frac{2}{\varepsilon}-2$ and $y=\frac{1}{\varepsilon}-1$ are integer numbers. Consider an instance with $x+y+1$ stubborn red agents, one stubborn blue agent, and two strategic blue agents. The topology and the placement of the stubborn agents is depicted in Figure~\ref{fig:pos-stubborn-general}. There are only three possible assignments depending on which pair of nodes (out of the three available) the two strategic blue agents occupy.

We claim that the only equilibrium assignment is the one where node $z_1$ is left empty with social welfare $\frac{1}{y+1} = \varepsilon$. First, observe that this assignment is indeed an equilibrium since no strategic agent has any incentive to deviate: node $z_1$ can give utility $\frac{1}{x+2}=\frac{\varepsilon}{2}$ to the agent occupying node $z_3$ and utility $0$ to the agent occupying node $z_2$. Since the two agents get utility $\frac{1}{y+1}=\varepsilon$ and $0$, respectively, none of them has any incentive to deviate. To verify the uniqueness of the equilibrium, observe that in the other two possible assignments there exists a strategic agent that can deviate to the empty node in order to increase her utility from $\frac{\varepsilon}{2}$ to $\varepsilon$ in case the strategic blue agents are connected, or from $\varepsilon$ to $1$ in case the strategic blue agents are not connected.

In contrast, the assignment according to which the strategic blue agents are connected to each other (by occupying nodes $z_1$ and $z_2$) is the optimal one with social welfare $1+\frac{\varepsilon}{2}$. Therefore, the price of stability is at least $\frac{1}{\varepsilon}+\frac{1}{2}$, which tends to infinity as $\varepsilon$ tends to zero.
\end{proof}

\begin{figure}[t]
\center
\includegraphics[scale=0.45]{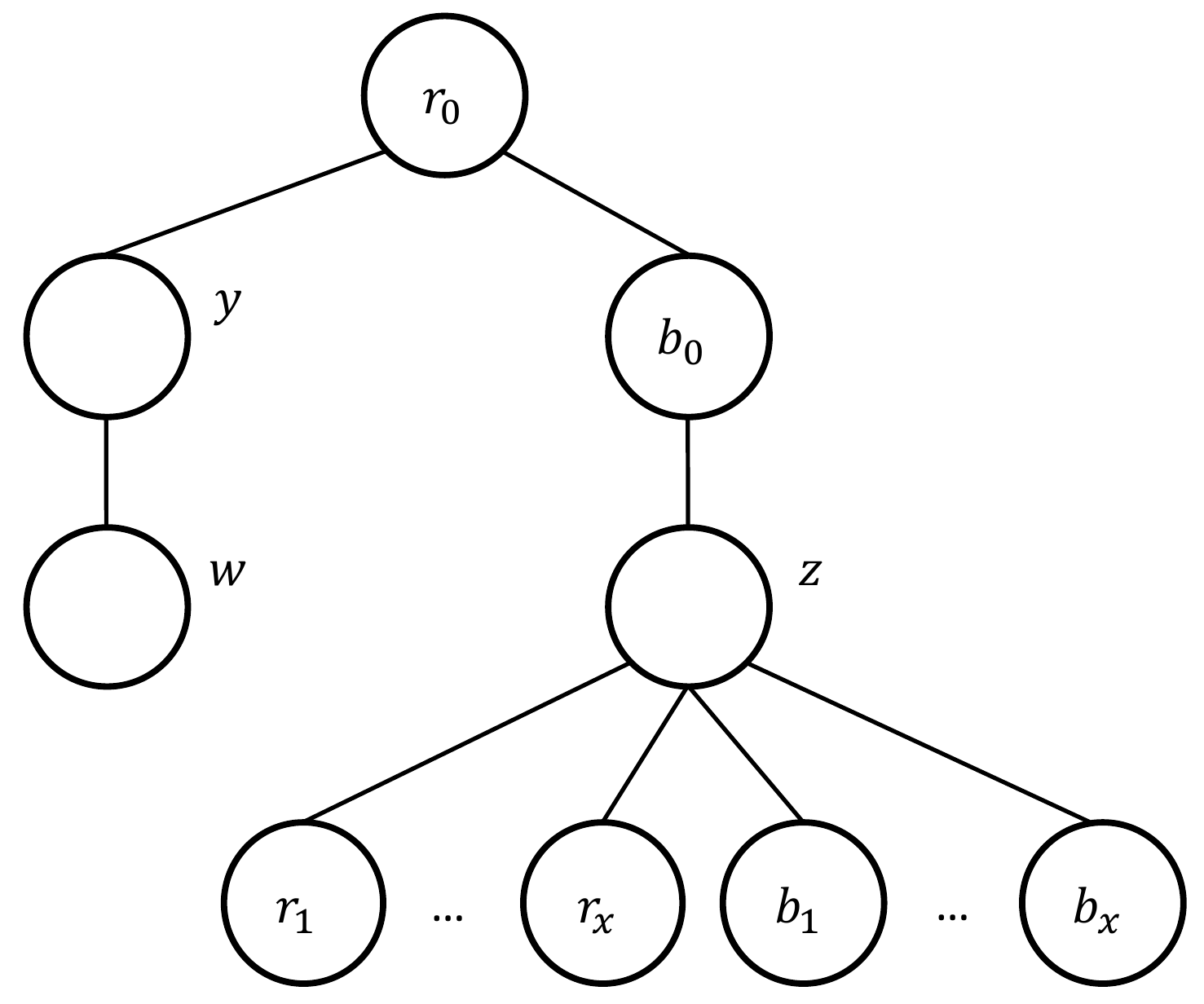}
\caption{The topology used in the proof of Lemma~\ref{lem:pos-stubborn-equal} for $k=2$.}
\label{fig:pos-stubborn-equal}
\end{figure}

\begin{lemma}\label{lem:pos-stubborn-equal}
For $k$-typed Schelling games, the price of stability
is at least $3$, for each even $k\geq 2$, if there is the same number of stubborn agents per type.
\end{lemma}

\begin{proof}
For simplicity, we will prove the lemma for $k=2$. Let $x \geq 1$ be a parameter and consider an instance with two types of agents (red and blue) such that there are $x+1$ stubborn red agents, $x+1$ stubborn blue agents, and two strategic blue agents. The topology and the placement of the stubborn agents are depicted in Figure~\ref{fig:pos-stubborn-equal}. There are only three possible assignments depending on which pair of nodes (out of the three available) the two strategic blue agents occupy.

We claim that the only equilibrium assignment is the one where node $y$ is left empty with social welfare $\frac{x+1}{2x+1}$. First, observe that this assignment is indeed an equilibrium since no strategic agent has any incentive to deviate: node $y$ can give utility $1/2$ to the agent occupying node $z$ and utility $0$ to the agent occupying node $w$. Since the two agents get utility $\frac{x+1}{2x+1} > 1/2$ and $0$, respectively, none of them has any incentive to deviate. To verify the uniqueness of the equilibrium, observe that in the other two possible assignments there exists a strategic agent that can deviate to the empty node in order to increase her utility from $1/2$ to $\frac{x+1}{2x+1} > 1/2$ in case the strategic blue agents are connected, or from $\frac{x+1}{2x+1}$ to $1$ in case the strategic blue agents are not connected.

In contrast, the assignment according to which the strategic blue agents are connected to each other (by occupying nodes $y$ and $w$) is the optimal one with social welfare $3/2$. Therefore, the price of stability is at least $\frac{3(2x+1)}{2(x+1)}$, which tends to $3$ as $x$ becomes arbitrarily large.

The bound can easily be extended to the case of $k$ types (for even $k\geq 2$) by replicating $k/2$ times the whole instance and connecting the topologies via an empty node.
\end{proof}

\begin{figure}[t]
\center
\includegraphics[scale=0.45]{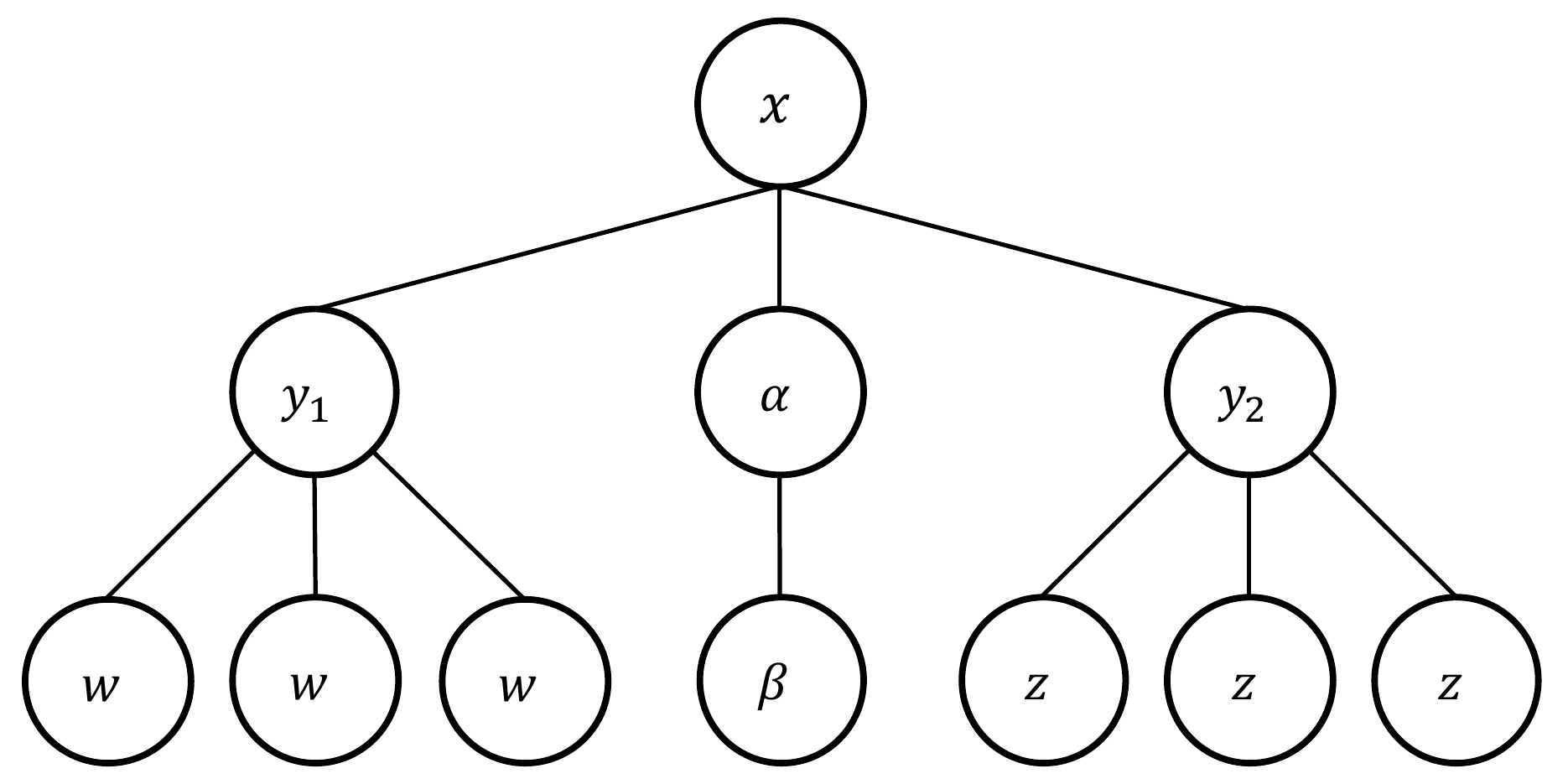}
\caption{The topology used in the proof of Lemma~\ref{lem:pos-strategic} for $k=2$.}
\label{fig:pos-strategic}
\end{figure}

\begin{lemma}\label{lem:pos-strategic}
For $k$-typed Schelling games, the price of stability
is at least $34/33$,
for each $k\geq 2$, even in the absence of stubborn agents.
\end{lemma}

\begin{proof}
For the sake of simplicity, we prove the lemma for $k=2$, for which the desired lower bound is $34/33$; 
we will discuss how to generalize our construction to $k > 2$ at the end of the proof.

Consider an instance with two types of agents (red and blue) such that there are five red and five 
blue agents; the topology is depicted in Figure~\ref{fig:pos-strategic}.

Let $\vv$ be the following assignment: node $x$, node $y_1$ and all three $w$-type nodes are occupied 
by red agents, while node $y_2$, all $z$-type nodes and node $\beta$ are occupied by blue agents. One 
can easily verify that $\vv$ is an equilibrium since no agent has any incentive to deviate to the 
empty node $\alpha$; the social welfare is $\SW(\vv) = 33/4$.

Let $\vv'$ be the following assignment: node $x$, node $y_1$ and all three $w$-type nodes are occupied 
by red agents, while node $y_2$, two of the $z$-type nodes, node $\alpha$ and node $\beta$ are 
occupied by blue agents. This is not an equilibrium assignment since the blue agent occupying 
$\alpha$ has utility $1/2$ and hence has an incentive to deviate to the empty $z$-type node 
in order to increase her utility to $1$. 
However, it achieves an improved social welfare of $\SW(\vv') = 34/4$.

In order to complete the proof, we need to argue that $\vv$ is an equilibrium
with the maximum social welfare. 
To this end, we establish some properties of equilibrium assignments.
\begin{itemize}
\item 
Node $x$ must be occupied. Assume otherwise that $x$ is left empty. If nodes $y_1$, $\alpha$ and $y_2$ 
are occupied by agents of the same type, then at least one of them will be connected to some agent of 
the other type, and therefore will have an incentive to deviate to $x$ in order to connect only to agents 
of the same type. Hence, without loss of generality (due to symmetry), at $y_1$, $\alpha$ and $y_2$ 
there are two red agents and one blue agent. Trivially, the blue agent cannot be connected to any 
red agents, since otherwise any such red agent would get zero utility and have an incentive to deviate 
to $x$ in order to increase her utility to $2/3$. Since there are four remaining blue agents, at least 
one of them must be connected to one of the two red agents occupying nodes at the second layer. Hence, 
this blue agent gets zero utility and has an incentive to deviate to $x$ in order to increase her utility 
to $1/3$.

\item 
Nodes $y_1$ and $y_2$ must be occupied. Assume otherwise that one of these nodes, say $y_1$, is left 
empty, while node $x$ is occupied by a red agent (without loss of generality). If all $w$-type nodes 
are occupied by agents of the same type, then all these agents get zero utility and have an incentive to 
deviate to $y_1$ in order to increase their utility to at least $2/3$. So, agents of both types must 
appear at the $w$-type nodes. But then, the red such agent has an incentive to deviate to $y_1$ in order 
to increase her utility from zero to at least $1/3$.

\item 
Agents of both types must appear at the nodes of the second layer. Assume otherwise that only agents 
of the same type appear at these nodes, while node $x$ is occupied by a red agent (without loss of 
generality). Let us further assume that $y_1$, $\alpha$ and $y_2$ are all occupied by blue agents. 
Then, since the empty node is one of those at the third layer, two of the blue agents occupying nodes 
$y_1$, $\alpha$ and $y_2$ have an incentive to deviate in order to increase their utility from strictly 
less than $1$ (since they are connected to the red agent occupying node $x$) to $1$. In case $y_1$, 
$\alpha$ and $y_2$ are all occupied by red agents, then all blue agents occupy nodes at the third 
layer, meaning that at least two of the red agents occupying $y_1$, $\alpha$ and $y_2$ get utility 
strictly less than $1$, and only one of them can be connected to the empty node. Hence, the other such 
red agent has an incentive to deviate to the empty node and increase her utility to $1$. Therefore, the 
empty node must be one of those at the second layer.

Since $y_1$ and $y_2$ are occupied, $\alpha$ has to be the empty node. If $\beta$ is occupied by a red 
agent, then this agent gets zero utility and has an incentive to deviate to $\alpha$ in order to connect 
to the red agent occupying $x$. Hence, $\beta$ must be occupied by a blue agent. If $y_1$ and $y_2$ 
are occupied by blue agents, then either one of them is connected only to red agents or both are 
connected to three red agents (including the one at $x$) and one blue agent. In any case, at least 
one of them has an incentive to deviate to $\alpha$ and increase her utility from $0$ or $1/4$ to $1/2$. 
So, nodes $y_1$ and $y_2$ must be occupied by red agents. But then, all blue agents occupy nodes at 
the third layer, get zero utility, and the four of them that are connected to the red agents occupying 
$y_1$ and $y_2$ have an incentive to deviate to $\alpha$ in order to increase their utility to $1/2$.

\item 
The type of agents that appears at node $x$ can appear only once more at nodes $y_1$, $\alpha$ or 
$y_2$. Assume otherwise that $x$ is occupied by a red agent (without loss of generality) and two nodes 
at the second layer, say $y_1$ and $\alpha$, are occupied by red agents as well; the case where $y_1$ and $y_2$ are occupied by red players is similar. By the discussion above, $y_2$ must then be occupied by a blue agent. 
Observe that since there are four 
remaining blue agents, one of them (agent $p$) has to be connected to one of the red agents 
occupying $y_1$ and $\alpha$. Trivially, none of the $z$-type nodes can be empty, since this would give an
incentive to $p$ to deviate there in order to increase her utility from zero to $1$. But then, this 
means that one of the $w$-type nodes or node $\beta$ is empty, thus giving an incentive to the red agent 
occupying $x$ to deviate in order to increase her utility from $2/3$ to $1$.
\end{itemize}

Given the above structural properties, there can only be two equilibria (and two more symmetric ones, 
produced by exchanging agents of different types):
\begin{itemize}
\item 
Nodes $x$ and $y_1$ are occupied by red agents, while nodes $\alpha$ and $y_2$ are occupied by blue 
agents; the assignment for the nodes of the third layer is then trivially defined. Such an equilibrium 
has social welfare $97/12$.
\item 
Nodes $x$ and $y_1$ are occupied by red agents, node $y_2$ is occupied a blue agent, and node 
$\alpha$ is empty; the assignment for the nodes of the third layer is trivially defined so that node 
$\beta$ is occupied by the last blue agent that gets zero utility. Such an equilibrium has social 
welfare $33/4$.
\end{itemize}
Hence, the second type of equilibrium assignment is the one with maximum social welfare, and the lower 
bound on the price of stability follows.

We can generalize the above instance to $k > 2$ agent types as follows. Let $G_0$ be the topology used in the above instance for two agent types. Now, consider an instance where the topology consists of $G_0$ and $k-2$ isolated nodes. There are five agents of type $T_1$, five agents of type $T_2$, and one agent per type $T_i$ for $i \in \{3,\dots,k\}$; the agents of type $T_1$ and $T_2$ correspond to the red and blue agents in the instance for $k=2$. 

Observe that the agents of type $T_i$ for $i \in \{3,\dots,k\}$ get zero utility in any possible assignment, since they are unique of their type. Consequently, even though there are many equilibrium assignments where these agents occupy nodes of $G_0$, none of these equilibria achieve higher utility than the ones where these agents occupy isolated nodes, and agents of types $T_1$ and $T_2$ occupy the nodes of $G_0$. Consequently, following the same reasoning as in the above instance for $k=2$, we can conclude that the price of stability is at least $34/33$.      
\end{proof}


\section{Variants and Extensions}\label{sec:extensions}
Throughout this paper, we focused on a setting where agents are classified into $k$ types
and their utilities are defined by the proportion of their friends among their neighbors.
In this section, we introduce three variants of this model and briefly discuss
some preliminary results; a more thorough 
investigation of these alternative models is left for future work.

\subsection{Schelling games with social networks}
In $k$-typed Schelling games,
the friendship relation is defined by types: an agent's set of friends
consists of all agents of the same type. One can also consider a more general friendship
relation, defined by an arbitrary undirected graph $\cal G$ with vertex set $N$, 
which we will refer to as the {\em social network}: the set
of friends of agent $i$ consists of all neighbors of $i$ in $\cal G$. 
We refer to the resulting class of games as {\em social Schelling games}. 

By definition, $k$-typed Schelling games form a subclass
of social Schelling games: a $k$-typed game corresponds to a social network consisting
of $k$ cliques. Hence, our next theorem implies Theorem~\ref{thm:typed-star-deg2}
in Section~\ref{sec:existence}. 

\begin{theorem}\label{thm:social-star-deg2}
Every social Schelling game
where the topology is a star or a graph of maximum degree $2$
admits at least one equilibrium assignment, which can be computed in polynomial time.
\end{theorem}

\begin{proof}
Consider a social Schelling game with a set of agents $N=R\cup S$, 
where $R$ is the set of strategic agents and $S$ is the set of stubborn agents, 
a topology $G=(V, E)$, a social network $\cal G$, and a function 
$\lambda:S\to V$ that describes the locations of the stubborn agents; for each $i\in R$, let 
$F(i)$ denote the set of nodes of $\cal G$ that are adjacent to $i$.

Suppose that $G$ is a star with center $v$.
Consider an assignment $\vv$ such that $v=v_i$ for some $i\in R\cup S$.
All strategic agents are indifferent among the leaves, so no agent in $R\setminus\{i\}$
has a beneficial deviation. Now, consider agent $i$. If $i$ is stubborn, she cannot deviate;
if $i$ is strategic,
she does not want to deviate, as any leaf node would give her zero utility.
Hence, $\vv$ is an equilibrium.

Now, suppose that $G=(V, E)$ is a graph of maximum degree $2$.
Our analysis for this case is inspired by Theorem~6 in the work of \citet{CLM18}.
For each $v\in V$, let $\deg(v)$ denote the degree of a vertex $v$ in $G$. 
Given an assignment $\vv$,
for each edge $e= \{v, w\}$, we define
$$
\phi(\vv, e)=
\begin{cases}
1      	&\text{if $w=v_i$, $v=v_j$ and $i\in F(j)$}\\
0	&\text{if $w=v_i$, $v=v_j$ and $i\not\in F(j)$}\\
\frac13	&\text{if $v$ or $w$ is unoccupied in $\vv$}.
\end{cases}
$$
Let $\Phi(\vv)=\sum_{e\in E}\phi(\vv, e)$. We claim that
$\Phi(\vv)$ is an ordinal potential function for our setting,
i.e., if an agent deviates to increase her utility, the potential function increases.

To see this, consider an assignment $\vv$ and an agent $i$ with $v_i=v$ that deviates to an empty node $w$;
denote the resulting assignment by $\vv'$.
Given an edge $e\in E$,
let
$$
\Delta(e)=\phi(\vv', e)-\phi(\vv, e).
$$
Also, for $z\in\{v, w\}$, let $\Delta(z)=\sum_{e:z\in e}\Delta(e)$.
Note that $i$'s move only changes the potential of edges incident to $v$
and $w$.
Hence, if $v$ and $w$ are not adjacent, we have
$\Phi(\vv')-\Phi(\vv)=\Delta(v)+\Delta(w)$.
We will now prove that $\Delta(v)+\Delta(w)>0$; if $v$ and $w$
are not adjacent, this establishes our claim; towards the end of the proof
we will explain how to handle the case $\{v, w\}\in E$.
We make the following observations.

\begin{itemize}
\item As no agent benefits from moving to an isolated node, it must be $\deg(w)>0$.

\item If $\deg(w)=1$, let $e_w\in E$ be the edge that is adjacent to $w$.
Since $w$ is empty in $\vv$, we have that $\phi(\vv, e_w)=\frac13$.
Since agent $i$ benefits from moving to $w$, we have that $\phi(\vv', e_w)=1$. 
Hence, $\Delta(e_w)=\frac23$ and, consequently, $\Delta(w)=\frac23$.

\item If $\deg(w)=2$, let $e_{w, 1}$ and $e_{w, 2}$ be the two edges incident to $w$.
Since $w$ is empty in $\vv$, we have that $\phi(\vv, e_{w, 1})=\phi(\vv, e_{w, 2})=\frac13$. 
Since agent $i$ benefits from moving to $w$, we have that $\phi(\vv', e_{w, 1})+\phi(\vv', e_{w, 2})\ge 1$. 
Hence, $\Delta(w)\ge \frac13$.

\item If $\deg(v)=0$ then by definition $\Delta(v)=0$.

\item If $\deg(v)=1$, let $e_v\in E$ be the edge that is incident to $v$.
Since $i$ benefits from moving away from $v$, we have that $\phi(\vv, e_v)\le \frac13$.
Since $v$ is left empty in $\vv'$, we have that $\phi(\vv', e_v)=\frac13$ and, consequently, $\Delta(v)\ge 0$.

\item If $\deg(v)=2$, let $e_{v, 1}$ and $e_{v, 2}$ be the two edges incident to $v$.
Since $v$ is left empty in $\vv'$, we have that $\phi(\vv', e_{v, 1})=\phi(\vv', e_{v, 2})=\frac13$. 
Since agent $i$ benefits from moving away from $v$, we have that $\phi(\vv, e_{v, 1})+\phi(\vv, e_{v, 2})\le 1$. 
Thus, $\Delta(v)\ge -\frac13$.
\end{itemize}
By the above observations, it follows that $\Delta(v)+\Delta(w)>0$ unless $\Delta(v)=-\frac13$ and $\Delta(w)=\frac13$.
However, this is impossible: $\Delta(v)=-\frac13$ only if in $\vv$ agent $i$
is adjacent to one friend and one enemy, and $\Delta(w)=\frac13$ only if in $\vv'$
agent $i$ is adjacent to one friend and one enemy; but in such a case, agent
$i$ would have no incentive to move, a contradiction.
This completes the analysis for when $\{v,w\} \notin E$.

Now, suppose that $v$ and $w$ are adjacent.
In this case we have that $\Phi(\vv')-\Phi(\vv) = \Delta(v)+\Delta(w) - \Delta(\{v, w\})$.
However, since $w$ is empty in $\vv$ and $v$ is empty in $\vv'$, 
it must be $\Delta(\{v, w\})=\frac13-\frac13=0$, and hence $\Delta(v)+\Delta(w)>0$
implies $\Phi(\vv')-\Phi(\vv)>0$ in this case as well.

Finally, note that the potential function takes values in the set
$\{\frac{\ell}{3}\mid \ell=0, \dots, 3|V|\}$, where $|V|$
is the number of nodes of the topology graph. Therefore, any best response dynamics
starting from an arbitrary initial configuration
converges to an equilibrium in $O(|V|)$ steps.
\end{proof}

Conversely, all our non-existence results (Theorem~\ref{thm:non-existence}), hardness
results (Theorems~\ref{thm:eq-hardness} and~\ref{thm:opt-hardness}) and lower bounds on the PoA
and PoS (Section~\ref{sec:poa-pos}) apply to social Schelling games
as well. In fact, maximizing the social welfare in social
Schelling games is NP-hard even if all agents are strategic (whereas our hardness reduction 
for $k$-typed games uses stubborn agents). Moreover, this hardness result
holds even if $G$ is a graph of maximum degree $2$, i.e., social welfare maximization may be hard
even when finding equilibria is easy.

\begin{theorem}\label{thm:social-opt-hardness}
Given a social Schelling game $I$ and a rational value $s$,
it is {\em NP}-complete to decide whether $I$ admits an assignment
with social welfare at least $s$.
The hardness result holds even if all agents are strategic
and even if $G$ is a graph of maximum degree $2$. 
\end{theorem}

\begin{proof}
It is immediate that our problem is in NP.
To show NP-hardness,
we will use a reduction from the {\sc Hamiltonian Cycle} (HC) problem.
An instance of HC is an undirected graph $H=(X, Y)$;
it is a yes-instance if and only if the vertices of this graph
can be ordered as $x_1, \dots, x_{|X|}$ so that $\{x_{|X|}, x_1\}\in E$
and for each $i \in [|X|-1]$ it holds that $\{x_i, x_{i+1}\}\in E$.

Given an instance $H=(X,Y)$ of HC, where $X$ is the set of nodes and $Y$ is the set of edges,
we construct an instance of our social welfare maximization problem as follows:
\begin{itemize}
\item For every node $v \in X$, we have a strategic agent $p_v$ with set of friends
$F(p_v)=\{p_z: \{z,v\} \in Y\}$.
\item The topology $G=(V,E)$ is a cycle consisting of $|X|$ nodes together with an isolated node $w$.
\end{itemize}
By construction, a social welfare of $|X|$ can be achieved if and only if
the agents can be assigned to the nodes of the cycle so that each of them
is adjacent to two friends; this is possible if and only if $H$ admits a Hamiltonian cycle.
\end{proof}

Identifying special classes of social Schelling games that allow for good upper bounds 
on the price of anarchy and the price of stability is an interesting research direction. 
We note that the upper bounds in Section~\ref{sec:poa-pos} only apply to $k$-typed instances 
with further restrictions on the structure of each type, so they cannot be extended 
to the social setting.

\subsection{Schelling games with enemy aversion}
In our model, if an agent is not adjacent to any friends, it does not matter how many enemies
she is adjacent to. This is also the case in fractional hedonic games:
agents are indifferent between being alone and being in coalitions consisting 
of their enemies. This assumption makes sense when the ``enemies'' of an agent are simply agents that do not 
contribute to her welfare. However, an agent may prefer being alone
to being in a group full of enemies. In the context of hedonic games, such preferences are modeled by 
{\em modified fractional hedonic games} \citep{O12,EFF16,BEI19}, where the utility of an agent in a coalition with 
$f$ friends and $e$ enemies is $\frac{f+1}{f+e+1}$, i.e., the agent herself is included
in the set of her friends. 

Many of our results extend to this definition of utility. 
For example, we can construct instances without equilibria even for $2$-typed games, using ideas similar to those in the reduction of Theorem~\ref{thm:eq-hardness}. 
Further, for $k$-typed games with a tree topology
and a constant number of types, equilibrium existence can be decided in polynomial time,  
by adapting the proof of Theorem~\ref{thm:existence-tree}.
However, it remains an open question if instances with no stubborn
agents always admit an equilibrium in this model.

\subsection{Schelling games with linear utilities}
Throughout the paper we assume that an agent's utility is determined by the fraction of her friends
among her neighbors. Alternatively, an agent may simply care about the number of friends
in her neighborhood 
or the difference between the number of friends $f_i$ and the number of enemies $e_i$;
more broadly, her utility may be an arbitrary linear function of $f_i$ and $e_i$
(in the context of hedonic games, this model corresponds to a subclass
of {\em additively separable hedonic games}; e.g., see \citep{AS16}). 
It turns out that games of this form are 
potential games and therefore have at least one equilibrium; furthermore, 
in the absence of stubborn agents there is always an equilibrium that is socially optimal.

\begin{theorem}\label{thm:additive-utilities}
Consider a variant of the (social) Schelling model where the utility of each agent $i$, 
who is adjacent to $f_i$ friends and $e_i$ enemies,
is $\alpha f_i-\beta e_i$ for some $\alpha, \beta\ge 0$. Then,   
every instance has an equilibrium assignment which 
can be computed in polynomial time.
Moreover, if no agent is stubborn, the price of stability is $1$.
\end{theorem}

\begin{proof}
Consider a game with a set of strategic agents $R$, 
a set of stubborn agents $S$, a topology $G=(V, E)$ and 
a friendship relation that is defined by a social network $\cal G$.
Fix non-negative constants $\alpha$ and $\beta$ such that the utility of an agent who is adjacent
to $f$ friends and $e$ enemies in the topology is given by $\alpha f- \beta e$. Our analysis is inspired by Proposition 2 in the work of \citet{BJ02}, showing that a Nash stable partition always exists in symmetric additively separable hedonic games. 

Let $N=R\cup S$ and $\vv$ be an assignment. For each $i\in N$ let
$\phi_i(\vv)=\alpha f_i(\vv)-\beta e_i(\vv)$, and $\Phi(\vv)=\sum_{i\in N}\phi_i(\vv)$.
We will argue that $\Phi$ is an ordinal potential function for our game.
Note that if all agents are strategic, $\Phi(\vv)$ is equal to the social welfare of $\vv$.
However, in general this is not the case: intuitively, $\Phi$ ascribes ``strategic'' utilities
to the stubborn agents.

Consider an assignment $\vv$ and an agent $i$ with $v_i=v$. Suppose that $i$
has a beneficial deviation from $v$ to another node $w\in V$, which is empty in $\vv$;
denote the resulting assignment by $\vv'$.
Suppose that agent $i$ has $f$ friends and $e$ enemies at $\vv$, and $f'$ friends and $e'$ enemies
at $\vv'$. Then, since the deviation is profitable, it holds that $\phi_i(\vv')-\phi_i(\vv) = \alpha(f'-f)-\beta(e'-e)>0$. We claim that
$\Phi(\vv')>\Phi(\vv)$.

Indeed, consider an agent $j\in N\setminus\{i\}$. 
If $j$ is a neighbor of $i$ in both $\vv$ and $\vv'$, or if 
$j$ is not a neighbor of $i$ in both $\vv$ and $\vv'$, 
then $\phi_j(\vv)=\phi_j(\vv')$.

Now, suppose that $j$ is adjacent to $i$ in $\vv$, but not in $\vv'$.
If $j$ is a friend of $i$, then $\phi_j(\vv')=\phi_j(\vv)-\alpha$, and
if $j$ is an enemy of $i$, then $\phi_j(\vv')=\phi_j(\vv)+\beta$.
Similarly, if $j$ is adjacent to $i$ in $\vv'$, but not in $\vv$, then
if $j$ is a friend of $i$, then $\phi_j(\vv')=\phi_j(\vv)+\alpha$, and
if $j$ is an enemy of $i$, then $\phi_j(\vv')=\phi_j(\vv)-\beta$.
Thus, the overall change in potential can be computed
as
\begin{align*}
\Phi(\vv')-\Phi(\vv) &= \phi_i(\vv') - \phi_i(\vv) 
-\alpha f + \beta e + \alpha f'- \beta e' \\
&= 2\bigg(\alpha(f'-f) - \beta(e'-e)\bigg)>0.
\end{align*}

It follows that, if the strategic agents follow the best response dynamics
starting from any initial configuration, they will converge to an equilibrium.
Moreover, the assignment that maximizes $\Phi$ is an equilibrium, 
so if all agents are strategic, this equilibrium maximizes the social welfare.
Note also that the function $\Phi$ takes values 
in the set $\{\alpha i - \beta j\mid 0\le i, j\le n^2\}$, where $n$
is the number of agents. Thus, any best response dynamics converges in $O(n^4)$
iterations.
\end{proof}

\section{Conclusions}\label{sec:conclusion}
In this paper, we investigated Schelling games on graphs, 
both from the perspective of equilibrium analysis and from the perspective
of social welfare. Concerning equilibrium existence, our positive
results are rather limited in scope: while an equilibrium always
exists for very simple topologies, such as stars and paths, 
it may fail to exist even if the topology does not contain cycles.
It would be interesting to obtain a complete characterization
of topologies that guarantee existence of equilibria.

For welfare maximization, a natural question is whether 
one can efficiently compute assignments with nearly optimal social welfare.
We note that our NP-hardness reductions are not approximation preserving, 
so they do not rule out this possibility.
Another interesting algorithmic question is whether the problem of computing equilibria
in $k$-typed games remains hard in the absence of stubborn agents; we conjecture
that this is indeed the case, but were unable to prove it.

\bibliographystyle{named}
\bibliography{schelling-bib}

\newpage
\appendix

\section{Proof of Theorem~\ref{thm:existence-tree}}
\newcommand{\ulb}{\check{u}}
\newcommand{\uub}{\hat{u}}
\newcommand{\Ulb}{\check{\mathbf{u}}}
\newcommand{\Uub}{\hat{\mathbf{u}}}
\newcommand{\nbr}{\mathbf{\mathbf{n}}}
\newcommand{\kbr}{\mathbf{\mathbf{k}}}
\newcommand{\true}{\text{\em true}}
\newcommand{\false}{\text{\em false}}
\newcommand{\tree}{\text{\em tree}}
\newcommand{\child}{\text{\em child}}
\newcommand{\blue}{\text{\em blue}}
\newcommand{\red}{\text{\em red}}
\newcommand{\ept}{\text{\em empty}}
\renewcommand{\top}{\text{\em top}}

For readability, we will present a polynomial-time algorithm that can decide whether an
equilibrium exists for instances with two agent types (red and blue) and no stubborn agents;
towards the end of the proof, we will explain how to extend it to instances
with a constant number of agent types that may contain stubborn agents.
Let $T_R$ denote the set of all red agents and
let $T_B$ denote the set of all blue agents.
Throughout the proof, we use the convention that a fraction of the form
$\frac{a}{b}$ evaluates to $0$ whenever $a=0$.

Consider an instance $I$ with $n$ agents and tree topology $G=(V, E)$.
Pick an arbitrary node $r$ to be the root of $G$.
Let $\tree(v)$ denote the set of descendants of $v$ (including $v$),
and let $\child(v)$ be the set of children of $v$.
Observe that the utility of a strategic agent takes values
in the set $\mathcal{U}=\{i/j: i\in[n], j\in [n], i\le j\}\cup\{0\}$;
note that $|\mathcal{U}|\le n^2$.

We use the following dynamic programming approach. For each
node $v \in V$, we fill out a table $\tau_v$, which contains an entry
$\tau_v(C, \nbr, \kbr, \Ulb, \Uub)$ for each tuple
$(C, \nbr, \kbr, \Ulb, \Uub)$, where
\begin{itemize}
\item $C\in\{\blue, \red, \ept \}$,
\item $\nbr=(n_B,n_R)\in[n]^2$,
\item $\kbr = (k_B, k_R)\in[n]^2$,
\item $\Ulb = (\ulb_B, \ulb_R, \ulb_{B^\dag}, \ulb_{R^\dag}) \in \mathcal{U}^4$, and
\item $\Uub = (\uub_B, \uub_R, \uub_{\top}) \in \mathcal{U}^3$.
\end{itemize}
Thus, the number of entries in each table is
$3\cdot n^4 \cdot |\mathcal{U}|^7$,
which is polynomial in the input size.

The value of each entry is either \emph{true} of \emph{false}.
Specifically,
$\tau_v(C, \nbr, \kbr, \Ulb, \Uub) = \true$
if and only if there exists an assignment
of a subset of agents to the nodes in $\tree(v)$
that satisfies the following conditions:
\begin{enumerate}
\item If $C=\ept$, then node $v$ is empty and otherwise it is assigned
      to an agent of color $C$. \label{enu:cdt_v}
\item Exactly $n_B$ nodes of $\tree(v)$ are assigned to blue agents, and
      exactly $n_R$ nodes of $\tree(v)$ are assigned to red agents. \label{enu:cdt_nbr}
\item Exactly $k_B$ nodes of $\child(v)$ are assigned to blue agents, and
      exactly $k_R$ nodes of $\child(v)$ are assigned to red agents.

\item Every blue agent in a node of $\child(v)$
      gets utility at least $\ulb_{B^\dag}$ and
      every red agent in a node of $\child(v)$
      gets utility at least $\ulb_{R^\dag}$. \label{enu:cdt_ulb_ch}
\item Every blue agent in a node of $\tree(v)\setminus ( \child(v)\cup \{v\})$
      gets utility at least $\ulb_B$ and
      every red agent in a node of $\tree(v)\setminus ( \child(v)\cup \{v\})$
      gets utility at least $\ulb_R$. \label{enu:cdt_ulb}

\item If a blue agent that is \emph{not} already in $\tree(v)$ moves
      to an empty node of $\tree(v)\setminus \{v\}$,
      her utility would be at most $\uub_B$, and
      if a red agent that is \emph{not} already in $\tree(v)$
      moves to an empty node of $\tree(v)\setminus \{v\}$,
      her utility would be at most $\uub_R$. \label{enu:cdt_uub}
\item If node $v$ is not empty, then the agent occupying $v$ can get utility at most $\uub_{\top}$
      by moving to an empty node of $\tree(v)\setminus \{v\}$. \label{enu:cdt_vub}

\item All agents in nodes of $\tree(v)\setminus \{v\}$ do not have
      an incentive to deviate to empty nodes of $\tree(v)\setminus \{v\}$. \label{enu:cdt_stb}
\end{enumerate}
Condition~\ref{enu:cdt_stb} directly relates to stability of $\tree(v)\setminus \{v\}$, whereas
conditions~\ref{enu:cdt_v}--\ref{enu:cdt_vub} are auxiliary, providing the necessary information
that we need in order to determine the stability of node $v$, and fill out the dynamic programming
table for the parent of $v$.

Consider the table $\tau_r$ at the root node $r$.
The game admits an equilibrium if and only if there exists
$(C, \nbr, \kbr, \Ulb, \Uub)$ such that
$n_B=|T_B|$, $n_R=|T_R|$,
$\tau_r(C, \nbr, \kbr, \Ulb, \Uub) = \true$
for the root node $r$ of $G$, and, moreover,
\begin{itemize}
\item
if $C = \blue$, then
$$\frac{k_B}{k_B + k_R} \geq \uub_{\top};$$
\item
if $C = \red$, then
$$\frac{k_R}{k_B + k_R} \geq \uub_{\top};$$
\item
if $C=\ept$, then for each $X\in\{R, B\}$ with $k_X>0$ it holds that
$$\frac{k_X}{k_B + k_R} \leq \ulb_X, \quad
  \frac{k_X-1}{k_B + k_R-1} \leq \ulb_{X^\dag}.$$
\end{itemize}
The first two conditions ensure that if the root node is not empty,
the agent in that node does not have an incentive to move to another
node of the tree, and the last condition ensures that if the root node
is empty, no agent has an incentive to deviate there (the exact form
of this condition depends on whether the potential deviator is located
in a child of $r$). Together with condition~\ref{enu:cdt_stb},
these conditions ensure that no agent wants to deviate.

The existence of a tuple $(C, \nbr, \kbr, \Ulb, \Uub)$ with these properties
can be decided in polynomial time by going through all entries of
$\tau_r$. It remains to show that $\tau_r$ can be filled in
in polynomial time.

Given $C\in\{\red, \blue, \ept\}$, we write
$\mathbbm{1}_B(C)=1$ if $C=\blue$ and $0$ otherwise;
similarly, $\mathbbm{1}_R(C)=1$ if $C=\red$ and $0$ otherwise, and
$\mathbbm{1}_E(E)=1$ if $C=\ept$ and $0$ otherwise.

We fill the tables in all nodes starting from the leaf nodes of $G$.
For every leaf node $v$, we have

\begin{align} \label{eq:Tv_leaf}
T_v(C, \nbr, \kbr, \Ulb, \Uub) 
=
\begin{cases}
\true, & \text{if } \nbr = \left(\mathbbm{1}_{B}(C),\mathbbm{1}_{R}(C) \right),
         \kbr = (0,0) \text{ and } \Uub = (0,0,0) \\
\false, & \text{otherwise}.
\end{cases}
\end{align}

Suppose now that for a node $w$ we have constructed the table $\tau_v$ for each $v \in \child(w)$.
We will construct $\tau_{w}$ using these tables as follows. Let $\child(w) = \{v_1, \dots, v_L \}$.
We create an intermediate table $\theta_w^{\ell}$ for each
$\ell \in \{0,1, \dots,L\}$. This table has an entry
$\theta_w^{\ell}(C, \nbr, \kbr, \Ulb, \Uub)$ for every tuple $(C, \nbr, \kbr, \Ulb, \Uub)$.
The entry $\theta_{w}^{\ell}(C, \nbr, \kbr, \Ulb, \Uub)$ is set to $\true$ if and only if
conditions~\ref{enu:cdt_v}--\ref{enu:cdt_stb}
hold for the subtree $\tree_\ell(w)$ obtained from $\tree(w)$ by deleting the subtrees rooted at
$v_{\ell+1}, \dots, v_L$.
Note that, by construction, we have
$\tau_w(C, \nbr, \kbr, \Ulb, \Uub) = \theta_w^{L}(C, \nbr, \kbr, \Ulb, \Uub)$.

We construct $\theta_w^{\ell}$ sequentially for $\ell=0, \dots, L$.
We can fill out $\theta_{w}^0$ using Equation~\eqref{eq:Tv_leaf}.
Next, suppose that we have filled out the first $\ell$ tables, i.e.,
$\theta_{w}^0, \dots, \theta_{w}^{\ell-1}$.
We combine $\theta_{w}^{\ell-1}$ and $\tau_{v_\ell}$ in order to build $\theta_{w}^\ell$
as follows: $\theta_{w}^\ell (C, \nbr, \kbr, \Ulb, \Uub) = \true$ if and only if
there exist a pair of tuples
$(C', \nbr', \kbr', \Ulb', \Uub')$ and $(C'', \nbr'', \kbr'', \Ulb'', \Uub'')$ such that
$\theta_{w}^{\ell-1}(C', \nbr', \kbr', \Ulb', \Uub') =
\tau_{v_\ell}(C'', \nbr'', \kbr'', \Ulb'', \Uub'') = \true$
and the following conditions hold:
\begin{enumerate}
\item $C' = C$.

\item $\nbr'' + \nbr' = \nbr$.

\item $\mathbbm{1}_B(C'') + k'_B = k_B$ and $\mathbbm{1}_R(C'') + k'_R = k_R$.

\item For each $X \in \{B,R\}$,
$$\ulb'_{X^\dag} \geq \ulb_{X^\dag}$$
so that the agents occupying nodes $v_1, \dots ,v_{\ell-1}$ have utility at least $\ulb_{X^\dag}$.
Additionally, if $C'' = \blue$, then
$$\frac{k''_B + \mathbbm{1}_B(C')}{k''_B + k''_R + (1-\mathbbm{1}_E(C'))} \geq \ulb_{B^\dag}$$
and, if $C'' = \red$, then
$$\frac{k''_R + \mathbbm{1}_R(C')}{k''_B + k''_R + (1-\mathbbm{1}_E(C'))} \geq \ulb_{R^\dag}$$
so that the agent occupying node $v_\ell$ has utility at least
$\ulb_{B^\dag}$ if she is blue or at least $\ulb_{R^\dag}$ if she is red.
Therefore, if these conditions hold, all agents occupying the first $\ell$ children of $w$
have utility at least $\ulb_{B^\dag}$ or $\ulb_{R^\dag}$, according to their type.

\item For each $X \in \{B,R\}$,
$$
\ulb'_X, \ulb''_X, \ulb''_{X^\dag} \geq \ulb_X,
$$
so that all agents of type $X$ occupying the nodes of $\tree_\ell(w)\setminus (\child(w) \cup \{w\})$
have utility at least $\ulb_X$.

\item For each $X\in\{B,R\}$,
$$\uub'_X \leq \uub_X, \uub''_X \le \uub_X$$
and, if $C'' = \ept$, then
$$
\frac{k''_X + \mathbbm{1}_X(C')}{k''_B + k''_R + (1-\mathbbm{1}_E(C'))} \leq \uub_X,
$$
so that the agents that do not occupy nodes of $\tree_\ell(w)$
have no incentive to deviate to any node in the first $\ell-1$ branches,
any node other than $v_\ell$ in the $\ell$-th branch, or node $v_\ell$.

\item If $C' = \blue$, then
$$\uub_{\top}' \leq \uub_{\top},\quad \uub''_B \leq \uub_{\top},$$
and, if $C'' = \ept$, then
$$\frac{k''_B}{k''_B + k''_R} \leq \uub_{\top},$$
so that the blue agent $i^*$ occupying node $w$ has utility at most $\uub_{\top}$ if she deviates
to a node in the first $\ell-1$ branches, a node in the $\ell$-th branch
(excluding node $v_\ell$), or node $v_\ell$. Similarly, if $C' = \red$, then
$$\uub_{\top}' \leq \uub_{\top}, \quad \uub''_R \leq \uub_{\top}$$
and, if $C'' = \ept$,
$$\frac{k''_R}{k''_B + k''_R} \leq \uub_{\top}.$$

\item
$\ulb''_B, \ulb''_{B^\dag} \geq \uub'_B$
so that blue agents occupying nodes in the $\ell$-th branch (excluding $v_\ell$)
have no incentive to deviate to any node in the first $\ell-1$ branches.
Also, if $C''=\ept$, then
$$\ulb''_B \geq \frac{k''_B + \mathbbm{1}_B(C')} {k''_B + k''_R + \mathbbm{1}_B(C')},$$
and if $k''_B>0$ then
$$\ulb''_{B^\dag} \geq \frac{k''_B + \mathbbm{1}_B(C')-1} {k''_B + k''_R + \mathbbm{1}_B(C')-1},$$
so that blue agents occupying nodes other than $v_\ell$ in the $\ell$-th branch
have no incentive to deviate to $v_\ell$.
Since $\tau_{v_\ell}(C'', \nbr'', \kbr'', \Ulb'', \Uub'')=\true$ means
that these agents already have no incentive to deviate to other empty nodes
in the $\ell$-th branch, now these agents have no incentive to deviate
to any empty node in $\tree_\ell(w)\setminus \{w\}$.
Further, if $C''=\blue$, then
$$\frac{k''_B + \mathbbm{1}_B(C')} {k''_B + k''_R + \mathbbm{1}_B(C')} \geq \uub'_B, \uub''_{\top},$$
so that if there is a blue agent at node $v_\ell$, she
has no incentive to deviate as well. Similar constraints must hold for red agents.

\item $\ulb'_B, \ulb'_{B^\dag}  \geq \uub''_B$ so that blue agents occupying nodes
in the first $\ell-1$ branches have no incentive to deviate to nodes in the $\ell$-th branch
(excluding node $v_\ell$). Additionally, if $C''=\ept$, then
$$\ulb'_B, \ulb'_{B^\dag} \geq \frac{k''_B + \mathbbm{1}_{B}(C')}{k''_B + k''_R},$$
so that blue agents in the first $\ell-1$ branches have no incentive to deviate to $v_\ell$
if it is empty. Similar constraints must hold for red agents as well.
\end{enumerate}
These constraints can be verified in polynomial time by checking each pair of entries of the tables
$\theta_{w}^{\ell-1}$ and $\tau_{v_\ell}$. This completes the proof for instances with
two agent types and no stubborn agents.

To extend the algorithm to instances with stubborn agents, we can set the entry
values of the table $\tau_v$ to $\false$ if $v$ is occupied by a stubborn agent
of a type other than $C$, and only consider possible deviations by strategic agents.  
The algorithm can trivially be extended to
instances with constant number of different agent types; the size of the tables
would scale exponentially with the number of types.
\hfill $\qed$


\end{document}